\documentclass[12pt, draftclsnofoot, onecolumn]{IEEEtran}

\usepackage{cite}
\usepackage{comment}
\usepackage{amsmath,amssymb,amsfonts}
\usepackage{mathtools}
\usepackage{bigints}
\usepackage{url}
\usepackage[printonlyused]{acronym}
\usepackage{nameref}
\usepackage{prettyref}
\usepackage{nicefrac}
\usepackage{afterpage}
\usepackage{algorithm}
\usepackage{algcompatible}
\usepackage{algpseudocode}
\usepackage[]{graphicx}
\graphicspath{{./figures/}}
\usepackage{eqparbox}
\usepackage[usenames,dvipsnames]{xcolor}
\usepackage{cite}
\usepackage[normalem]{ulem}
\usepackage{soul}
\usepackage[draft,nomargin,marginclue,footnote,silent]{fixme}
\usepackage[font=small]{caption}
\usepackage{placeins}
\usepackage{amsthm}
\usepackage{array}

\usepackage[caption=false]{subfig}

\newcommand{\Ptx}{p}
\newcommand{\vPtx}{\vec{p}}
\newcommand{\vB}{\vec{B}}
\newcommand{\wepsilon}{w^{\varepsilon}}
\newcommand{\vPtxprim}{\vec{p'}}
\newcommand{\vgammabar}{\vec{\gammabar}}

\newcommand{\gammabar}{\bar{\gamma}}
\newcommand{\kts}{k_{\mathrm{ts}}}
\newcommand{\M}{\mathcal{M}}

\newcommand{\Kd}{\mathcal{K}}

\newcommand{\Deltamin}{\Delta_{\mathrm{min}}}
\newcommand{\deltaEpsilon}{\Delta_{\varepsilon}}

\newcommand{\deltaPtx}{\Delta p}
\newcommand{\ra}{r_{\mathrm{a}}}
\newcommand{\deltaPtxmin}{\Delta p_{\mathrm{min}}}
\newcommand{\Ptxmax}{p_{\mathrm{max}}}
\newcommand{\Ptxmin}{p_{\mathrm{min}}}

\newcommand{\Ptxn}{{p}_n}

\newcommand{\EE}{\mathbb{E}}
\newtheorem{theorem}{Theorem}

\def\minplus{{$(\min, +)\,$}} 
\def\S{{\cal S}}
\def\A{{\cal A}}
\def\D{{\cal D}}

\def\K{{\cal K}}

\def\M{{\mathcal M}}

\def\mx{{$(\min, \times)$}}

\def\calS{\cal{S}}

\def\calSc{\mathcal{S}_c}
\def\L{\mathbb{L}}

\allowdisplaybreaks

\begin{document}

\title{  Bound-Based Power Optimization for Multi-Hop Heterogeneous Wireless Industrial Networks Under Statistical Delay Constraints} 
\author{
    \IEEEauthorblockN{Neda Petreska\IEEEauthorrefmark{1}, Hussein Al-Zubaidy\IEEEauthorrefmark{4}, Rudi Knorr\IEEEauthorrefmark{1}, James Gross\IEEEauthorrefmark{4}} \\
    \IEEEauthorblockA{\IEEEauthorrefmark{1}Fraunhofer Institute for Embedded Systems and Communication Technologies ESK
    \\\{neda.petreska, rudi.knorr\}@esk.fraunhofer.de} \\
    \IEEEauthorblockA{\IEEEauthorrefmark{4}School of Electrical Engineering, KTH Royal Institute of Technology
    \\{hzubaidy@kth.se, james.gross@ee.kth.se}
}
}
\IEEEpeerreviewmaketitle
\maketitle
\begin{abstract}
The noticeably increased deployment of wireless networks
for battery-limited industrial applications in recent years highlights the
need for tractable performance analysis methodologies as well as efficient 
QoS-aware transmit power management schemes. 
In this work, we seek to combine several important
aspects of such networks, i.e., multi-hop  connectivity, channel heterogeneity and the queuing effect, in order to address these needs.
We design delay-bound-based algorithms for transmit power minimization and network lifetime maximization of multi-hop heterogeneous wireless networks using our previously developed stochastic network calculus approach  for performance analysis of a cascade of  buffered wireless  fading channels. 
Our analysis shows an overall transmit power saving of up to 95\%  compared to a fixed power allocation scheme when using a service model in terms of the Shannon capacity limit. 
For a more realistic set-up, we  evaluate the performance of the suggested algorithm in a  WirelessHART network, which is a widely  used communication standard for process automation and other industrial applications.
We find that  link heterogeneity can significantly reduce  network
lifetime  when no efficient power management is applied. 
Moreover, we show, using extensive simulation study, that the proposed bound-based power allocation performs reasonably well compared to the real optimum, especially in the case of WirelessHART networks.
\end{abstract}
\section{Introduction}
\label{sec:intro}
In recent years, wireless networking solutions are increasingly being deployed to many new domains such as vehicular networks, machine-to-machine (M2M) communication, home automation, industrial settings and the smart grid. 
Applications in these areas often require novel combinations of delay and reliability constraints, while on the other hand relying on battery-driven wireless systems.
One area where these aspects are especially important is the area of wireless industrial networks, in particular, process automation. 
Process automation comprises the area of process sensing, control and diagnostics. 
Target application areas can be found, for example, in refineries, food and chemical industries. 
Typical process automation applications have Quality-of-Service (QoS) demands with deadlines in the order of hundreds of milliseconds and maximum outage probabilities (with respect to the deadlines) in the order of $10^{-3}$ to $10^{-4}$ \cite{zvei}.
At the same time, these applications correspond to factory sites which can span over quite a wide range of distances, varying from few meters up to few kilometers. 
Battery-driven wireless sensors and actuators are mainly applied in these scenarios due to their flexible placement possibilities. 
Of particular relevance to this work is the tendency of modern industrial solutions to deploy \emph{multi-hop topologies} in order to bridge larger
distances without necessarily shortening the lifetime of battery-powered
nodes. This poses a significant challenge to existing research efforts to determine the optimal transmit power allocation in such networks. 
Due to the fact that transmit power is one of the main energy consumers in a wireless device~\cite{berry}, an adaptive transmit power management has the potential to improve battery lifetime in addition to other techniques, such as hardware design, load balancing and transceiver state management~\cite{zheng_ondemand}. 
 
Nevertheless, when it comes to dependable industrial and machine-to-machine applications, transmit power management is challenging as it has to factor in not only the time-varying transmission rate of the wireless channel due to  fading, but also latency and reliability constraints. 
This necessitates a trade-off between power consumption at each node and the physical channel transmit rate, which potentially (due to time variability) leads to a queue build-up at that node. 
Although transmit power management under QoS requirements has been addressed for single-hop communications~\cite{berry,zafer_modiano_calculus,tang_zhang_singlehop,kandukuri_boyd}, power management under  stochastic queuing constraints for heterogeneous multi-hop wireless networks remains, to date, an open problem. Solving this problem is the focus of this work. In order to do so, we first need to express, mathematically, the queuing performance of the wireless network in terms of the transmit power, the fading environment and the used transmission technology/protocols, then solve this expression for the optimal power allocation under the given latency and reliability constraints.  
Unfortunately, an \textit{exact} expression for the performance of heterogeneous multi-hop wireless networks is not tractable so far. 
Nevertheless, recent advancements in stochastic network calculus~\cite{alzubaidy_ton,icc15} enable us to derive probabilistic delay bounds for heterogeneous networks.

Based on this insight, in this work we develop an analytical model for the performance of heterogeneous multi-hop wireless networks. Using an approach that we developed previously~\cite{icc15},  we provide a tractable expression for the end-to-end probabilistic delay bound. 
We then propose a bound-based optimal power allocation algorithm that minimizes transmit power while maintaining a bound on the delay performance.   
It is worth noting that optimizing the bound may result in a different operating point (in terms of power allocation per link, for instance) than the true system optimum (for which no tractable analytical model exists).
Taking this into account, we intend to answer the following  essential questions: 
\begin{enumerate}
\item Is it feasible to provide a bound-based optimal power allocation along a multi-hop path  of heterogeneous wireless fading channels under statistical delay constraints? If so, under what conditions?
\item How would such power allocation scheme look like? 
\item How good is such bound-based optimal power allocation scheme and how well does it perform compared to the real optimum?
\end{enumerate} 

In the following, we study the structure of the derived end-to-end delay bound and prove some of its important properties for optimization such as convexity and monotonicity of the `delay kernel,' which we define in Section \ref{sec:recursive}. This addresses the first  question above and  enables us to develop  a power minimization algorithm for optimal power allocation in a wireless multi-hop path under statistical end-to-end delay constraints. 

To address the second question, we use the derived closed-form solution from \cite{icc15} for the end-to-end delay bound for heterogeneous multi-hop wireless networks to develop two optimization algorithms: (i) a \emph{bound-based power-minimization algorithm} and (ii) a \emph{bound-based network lifetime maximization algorithm}. The first algorithm determines the minimum transmit power required, at each hop, such that a given statistical delay  bound (which encompasses both latency and reliability) is not violated. The second algorithm determines the allocation of transmit power among all transmitters along the respective path such that for given initial energy levels (e.g., battery charge per intermediate transmitter), the network lifetime is maximized given that the designated statistical delay constraint is not violated. 
While the first algorithm is suitable for energy efficient networks where energy levels can be replenished (i.e., recharging their batteries is possible), the second is more useful for wireless sensor network applications, especially, for remote area deployment where battery recharging is not possible. This motivated the development of the two separate algorithms.
Although the two algorithms may result in the same power allocation scheme for the homogeneous case (assuming identical fading distribution and same initial energy levels for all hops), this is obviously not the case for heterogeneous networks that we are interested in. 


To address the third question, we first recognize that optimizing a statistical bound on the delay may differ from optimizing the exact expression for the delay violation probability. Nevertheless, due to the lack of such an exact expression for now (and for the foreseeable future, due  to the intractability of the analysis in this case), we opt to use the (more tractable) statistical delay bound for our analysis and optimization instead. In this case, it is important to quantify the optimality gap in energy efficiency of the bound-based algorithms compared to the real optimum, which we obtain by simulations. In Section \ref{sec:num_res} we provide an extensive numerical study to address this point.
%
We show that in many cases of link heterogeneity our algorithm provides a sufficient estimate on the optimal transmit power per node, at the same time avoiding the need for extensive and time-consuming system optimization using simulations. 

In the remainder of this section, we present a literature survey of related work, and then list the main contributions of this work. We first discuss related work with respect to general end-to-end transmit power management schemes, then we discuss related work that deals with end-to-end queuing performance.   


\subsection{Transmit Power Management}
Power management under \textit{simplified} end-to-end throughput constraints has been often addressed before~\cite{kozat,cruz_power_multi,katsenou,banerjee,julian_boyd}. 
However, none of these works considers queuing effects.
For instance, \cite{kozat} describes a cross-layer design framework, minimizing the total transmit power subject to a minimal end-to-end payload rate valid for all links and maximal bit error rate (BER) requirements per session. 
The authors use heuristics to determine the transmit power per node.
\cite{cruz_power_multi} minimize the total average transmit power under the constraint of providing a minimum average data rate per link. 
The authors propose an algorithm for optimal link scheduling and power control policy. 
They also extend this to a routing algorithm, which uses the algorithms output as routing metric. 
They show that the optimal power policy chooses one of two actions, transmitting at peak power or not transmitting at all. 
Transmit power control in multi-hop networks with respect to the best possible video quality at the receiver is presented in~\cite{katsenou}. 
For this purpose, the authors maximize the peak signal-to-noise ratio and minimize the end-to-end video distortion, which is a function of the end-to-end bit error probability. 
Results show that power control does not degrade video quality significantly. 

With respect to power management and end-to-end multi-hop performance, two works are perhaps closest to our contribution.
First, \cite{tang_zhang} presents a tradeoff between the average transmit power and a corresponding queuing-delay bound for a multiuser cellular network, multi-hop and point-to-point communication. 
The authors propose a resource allocation scheme to minimize power consumption subject to statistical delay QoS, given as a queue-length decay rate, jointly determined from the effective bandwidth of the arrival traffic and the effective capacity of the wireless channel. 
The numerical analysis shows that it is possible to achieve stringent QoS guarantee with little power increase compared to the power needed for loose delay constraints. 
However, the discussed multi-hop scenario assumes an amplify-and-forward scheme, and therefore does not consider queuing at the intermediate nodes. 
Second, in~\cite{neely_modiano} the authors present a joint routing and power allocation policy for a wireless multi-hop network with time-varying channels that stabilizes the system and provides bounded average delay guarantees. 
The optimal power allocation in both proposed schemes, defined as distributed and centralized control algorithm, is determined under pre-defined stability condition, i.e., for input rates which are strictly inside the network capacity region. 
It is shown that the derived delay bounds grow asymptotically in the size of the network and a parameter that describes the distance between the arrival rates and the capacity region boundary. 
The numerical results illustrate the advantage of exploiting channel state and queue backlog. 
The work, however, covers queuing by focusing mainly on a stability condition and does not consider quantiles on the end-to-end delay. 

\subsection{Queuing Analysis of Wireless Networks}
From a queuing-theoretic perspective, significant problems arise when trying to characterize such quantiles on the end-to-end delay performance of a wireless multi-hop network.
Classical models for queuing networks typically only allow the analysis of the average delay. 
In contrast, the theory of network calculus enables an analysis of delay quantiles via bounds on the arrival and service rather than focusing on the average behaviour. 
In particular, stochastic network calculus~\cite{snc_book} has shown to be especially useful for characterizing traffic arrivals and network service of wireless multi-hop networks.
While there are many works addressing performance guarantees over wireless fading channels using this theoretical framework~\cite{jiang_emstad,fidler_mgf,fidler_wowmom,fidler_globecom,florin_throughput,florin_capacity,florin_infocom,ciucu_sigmetrics}, none of them resolve the question on end-to-end delay bounds over heterogeneous fading channels. 
Attempts in that direction are presented in~\cite{florin_throughput} and \cite{fidler_mgf}. 
However, while the first one does not provide a closed-form expression for the end-to-end service curve for wireless fading channels, the complexity of the MGF-based framework presented in~\cite{fidler_mgf} grows very fast when considering a heterogeneous multi-hop path and results in mathematically intractable expressions. 
Furthermore, the usage of the Gilbert-Elliott two-state channel model in~\cite{fidler_globecom} limits the accuracy of the fading model description.

However, a recently developed theoretical framework has enabled a new analytical toolset for performance analysis of wireless multi-hop fading channels \cite{alzubaidy_ton}. 
By means of $(\mathrm{min},\times)$-calculus, bounds on the delay and the backlog are expressed in terms of fading channel gain distribution, working directly in the so called SNR domain. 
In this domain, multi-hop descriptions of fading channels become mathematically tractable. 
Based on this work, we have made first attempts to determine the minimal required SNR on a single link in order to meet pre-defined statistical delay requirements~\cite{itc14}.
However,~\cite{alzubaidy_ton} addresses only independent and identically distributed (i.i.d.) wireless channel gains, which limits the applicability of the results to general scenarios. 

\subsection{Contributions and Paper Organization}
This paper builds on the  framework that we proposed in~\cite{itc14} and extends it to multi-hop buffered wireless links with heterogeneously distributed channel gains. Motivated by the discussion above, and the need for energy-efficient heterogeneous multi-hop wireless networks for future industrial applications, we present in this paper  the following main contributions: 
\begin{itemize}
	\item Based on the previously derived closed-form expression for the end-to-end statistical delay bound for a multi-hop path consisting of independent, but heterogeneously distributed channel gains, presented in \cite{icc15}, an iterative bound-based power-minimization algorithm for wireless multi-hop heterogeneous networks is developed. We present two variations of the algorithm: (i) minimizing the total transmit power along the path, and (ii) maximizing the network lifetime. A numerical evaluation of the performance of the two variants is performed.
	\item A proof for the convexity of the delay bound is provided and an evaluation of the  bound-based power minimization compared to simulation-based power optimization, i.e., using the exact delay process instead of a delay bound, is presented.
\end{itemize}
%

We conduct  our analysis and evaluation using two different channel capacity models, (i) an `ideal' Shannon-capacity-based model, and (ii) a more realistic WirelessHART (IEEE 802.15.4)-based link model.  
The  results that we obtained for power gain (using Shannon-based capacity) and network lifetime extension (assuming an IEEE 802.15.4-based link capacity) offer significant insights into multi-hop network design, considering heterogeneity of both channel gain as well as battery charge-state (i.e., initial energy level). 

The remaining paper is organized as follows:
Section~\ref{sec:preliminaries} presents the system model and the problem statement. 
In Section~\ref{sec:recursive} we present the closed-form of the end-to-end delay bound for heterogeneous wireless networks.
These results are the basis for the power minimization algorithms presented in Section~\ref{sec:algorithm}. 
We discuss the numerical evaluation of the algorithms in Section~\ref{sec:num_res}. 
Finally, Section~\ref{sec:conclusion} concludes the paper.
\section{System Model and Problem Statement}
\label{sec:preliminaries}

We consider the communication between a source node $r$ and a destination node $d$ within a multi-hop wireless network (see Fig.~\ref{fig:multihop_path}). 
Let the multi-hop path in question, illustrated with the solid lines in Fig.~\ref{fig:multihop_path}, be given with an ordered set of buffered links, i.e., \mbox{$\mathbb{L} = \{1,..,N\}$}, where $N$ is the number of links constituting path $\mathbb{L}$.
\begin{figure}
\centering
\captionsetup{justification=centering}
  \includegraphics[scale=0.4]{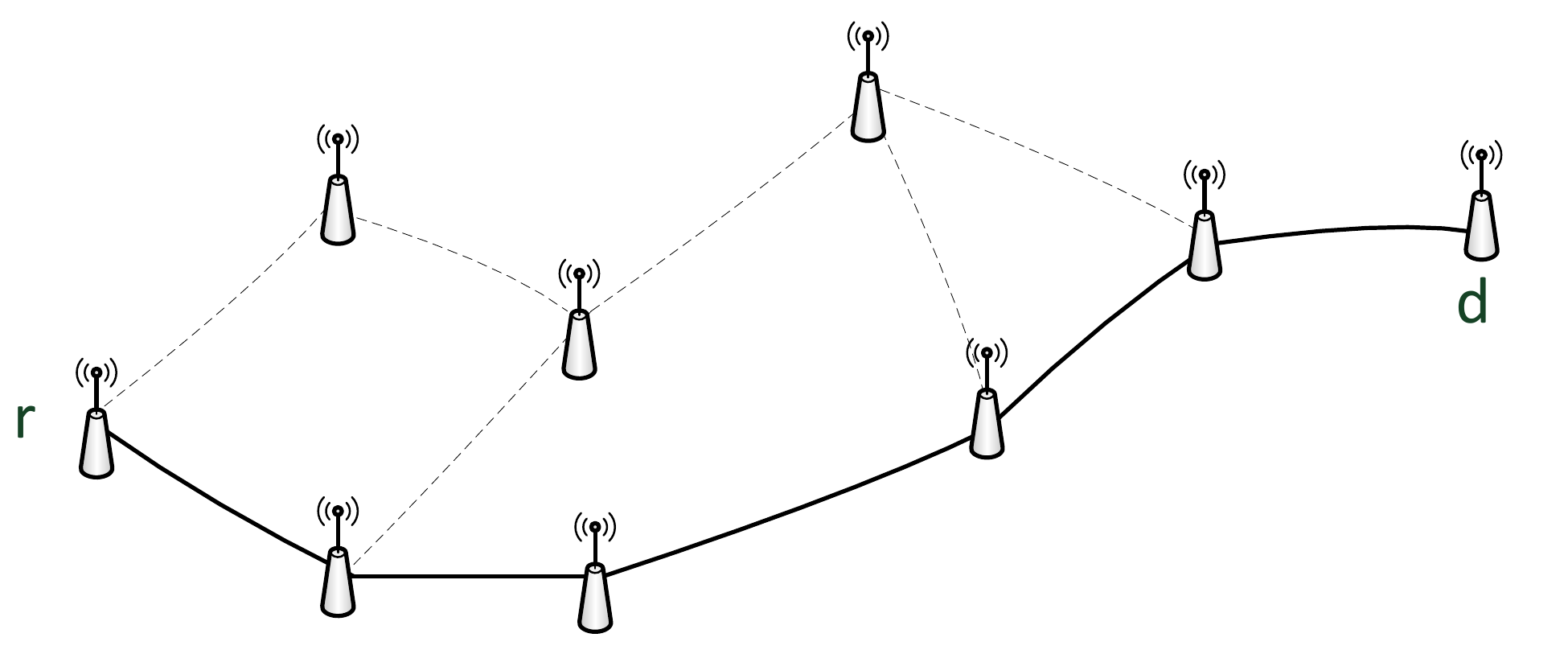}
\caption{Illustration of the system model}
\label{fig:multihop_path}
\end{figure} 
We assume a time-slotted system where every link is assigned a time slot of fixed-length $T$. 
The nodes propagate the packets along the path every time they are allocated a transmission time slot, according to a given scheduling algorithm.
As soon as a packet reaches the destination node, it is passed to its application layer without any additional delay. 

Each wireless link $n\in\mathbb{L}$ is assumed to be a block-fading channel with an average channel gain $|{\bar{h}}_n^2|$. 
This means, the random instantaneous channel gain $h_{i,n}^2$ of link $n$ at time slot~$i$ remains constant within the time interval $T$, but varies independently from slot to slot\footnote{This assumption is valid for a low mobility  environment where the channel gain remains constant within one transmission interval. Furthermore, frequency-hopping channel is one example where the channel gain varies independently from one transmission interval to the other \cite{slow_fading}.}.
Different links are assumed to have statistically independent channel gains while in general we consider heterogeneously distributed channel gains. Moreover, the instantaneous channel gain consists of two components: the instantaneous fading component $h_f^2$ and the constant path-loss $h_p^2$, the latter depending on the distance between the nodes and the path-loss exponent, i.e., $h_{i,n}^2={h_{f}^2}_{i,n}\cdot {h_{p}^2}_{n}$.
Although interference is not considered in this paper explicitly, we provide an approach, using the concept of `leftover service curve', to model interference from cross-flows in the network.
Together with the transmit power setting $p_{i,n}$ and the noise power $\sigma^2$, this yields the instantaneous SNR of link $n$ in time slot $i$ as:
\begin{equation}
\gamma_{i,n} = \frac{p_{i,n} \cdot h_{i,n}^2}{\sigma^2}
\label{eq:SNR_definition}
\end{equation} 
Given the instantaneous SNR $\gamma_{i,n}$, the resulting service per link $n$ is given by the link's capacity at time slot $i$.
We denote the capacity of link $n$ during time slot $i$ by the function $C \cdot \log_2 g(\gamma_{i,n})$, where $C$ is the number of symbols per time slot.
 We assume that transmitter $n$ knows only the channel state information  for the $n^{\rm th}$ channel  (corresponding to  one of the (up to) $N$ transmission slots within a transmission frame), but not for the other $N-1$ channels.
Furthermore, we define $\vgammabar$ and $\vPtx$ as the vectors of average SNRs and transmit power of the links, i.e., nodes along the path, respectively.

At the application layer, we consider a monitoring process generating a measurement value at a regular interval. 
Therefore, as a model for the arrival flow we consider packets of size $\ra$ bits arriving per interval $R_a \cdot T$ with $R_a \in \mathbb{N}$.
However, the application has strict latency and reliability constraints. 
These are modeled by the QoS pair $\left\{\omega^{\varepsilon}, \varepsilon\right\}$ where $\omega^{\varepsilon}$ represents a maximum tolerable delay that can be violated at most with probability $\varepsilon$.
This delay target includes all processing steps below the application layer, therefore including also any queuing delay along the multi-hop path.

In this paper, we are interested in the trade-off between the energy consumption of the network - mainly driven by the transmit power per node $p_{i,n}$ - and the resulting delay and delay violation probability.
In this context, we are  particularly interested in the minimization of the sum of the transmit power along the path (from now onward referred to as \emph{total transmit power}) under a given QoS pair $\left\{\omega^{\varepsilon}, \varepsilon\right\}$, and  in the maximization of the network lifetime given a battery state vector $\vB$ consisting of battery states per node\footnote{In this model, node $n$ refers to the transmitter node preceding link $n$.}, $B_{n}$, as well as an energy consumption model, under a given QoS pair $\left\{\omega^{\varepsilon}, \varepsilon\right\}$. 

In the following, we will develop algorithms that determine the corresponding system configurations (i.e., transmit power settings per node) under the desired QoS constraints. 
These algorithms are based on an analytical expression of the network performance (i.e., latency and reliability) in terms of the underlying physical channel properties, which includes the power allocation scheme. Hence, the first step in our efforts is to develop an analytical model of the end-to-end delay performance in a wireless multi-hop network in terms of transmit power allocation, which we address  in the next section. Then, we develop two  algorithms for total power minimization and for network lifetime maximization in Section~\ref{sec:algorithm}.

\section{End-to-End Delay Bound over Heterogeneous Links} 
\label{sec:recursive}

In this section, we develop end-to-end performance bounds based on stochastic network calculus for heterogeneous, multi-hop communication paths.
For wireless fading channels, end-to-end probabilistic bounds have only been obtained for concatenated i.i.d. service processes (i.e., for multi-hop wireless links all having independent and identically distributed fading processes) \cite{alzubaidy_ton}.
In order to apply these results to heterogeneous networks, we generalize the available results to arbitrarily distributed random service processes. 
For reader's benefit, we first recap some network calculus basics before presenting our theoretical results.

\subsection{Stochastic Network Calculus}
\label{sec:snc_basics}

Stochastic network calculus considers queuing systems and networks of systems with stochastic arrival and departure processes, where the bivariate functions $A(\tau,t)$, $D(\tau,t)$ and $S(\tau,t)$, for any $0 \le \tau \le t$, denote the \textit{cumulative} arrivals to the system, departures from the system, and service offered by the system, respectively, in the interval $[\tau,t)$.  
Recall that we consider a discrete time model, where time slots have a duration $T$ and $i \geq 0$ denotes the index of the respective time-slot. Hence, $t,\tau,i \in \mathbb{Z}$.

A lossless system with service process $S(\tau,t)$ satisfies the input/output relationship $D(0,t) \geq A \otimes S \left(0,t\right)$, where $\otimes$ is the \minplus~convolution operator defined as
\begin{equation}
 x \otimes y \left(\tau,t\right) = \inf_{\tau \leq u \leq t} \left\{ x(\tau,u) + y (u,t) \right\} \; .
 \label{eq:convolution}
 \end{equation}
 
In this approach, we are generally interested in probabilistic  bounds of the form $\mathrm{Pr}\left[ W(t) > w^{\varepsilon} \right] \leq \varepsilon$, which is also known as the \textit{violation probability} for a target delay $w^{\varepsilon}$, under the following system stability condition:
\begin{equation}
\label{eq:stability}
\lim_{t\to \infty} {\frac{A(0,t)}{t}} < \lim_{t\to \infty} {\frac{S(0,t)}{t}}.
\end{equation}

Modeling wireless links in the context of network calculus however is  not  a trivial task. 
A particular difficulty arises when we seek to obtain a stochastic characterization of the cumulative service process of a wireless fading channel, as also witnessed in the context of the effective service capacity of wireless systems \cite{wu}. 
A promising, recent approach for wireless networks has been proposed in \cite{alzubaidy_ton} where the queuing behavior is analyzed directly in the ``domain'' of channel variations instead of the bit domain \cite{jiang:servermodel,fidler_globecom,lee_jindal,mahmood:mimo,fidler_mgf,wu}.
This can be interpreted as the \textit{SNR domain} (thinking of bits as ``SNR demands'' that reside in the system until these demands can be met by the channel).

To start with, the cumulative arrival, service, and departure processes in the bit domain, i.e., $A$, $D$, and $S$, are related to their SNR domain counterparts (represented in the following by calligraphic capital letters $\A$, $\D$, and $\S$) respectively, through the exponential function.
Thus, we have \mbox{$\mathcal{A}(\tau,t) \triangleq e^{A(\tau,t)}$}, $\mathcal{D}(\tau,t) \triangleq e^{D(\tau,t)}$, and $\mathcal{S}(\tau,t) \triangleq e^{S(\tau,t)}$.
Due to the exponential function, these cumulative processes become products of the increments in the bit domain. In the following, we will assume $\mathcal{A}\left(\tau,t\right)$ and $\mathcal{S}\left(\tau,t\right)$ to have stationary and independent increments.
We denote them by $\alpha$ for the arrivals (in SNR domain) and $g\left(\gamma\right)$ for the service.
For instance, assuming a single-hop wireless system with a point-to-point channel, where the SNR domain is related to the bit domain through the well-known Shannon capacity expression, as follows 
\begin{equation}
s_{i} =  \log (g\left(\gamma_i\right)) = C \log_2\left( 1+ \gamma_i \right),
\label{eq:Shannon}
\end{equation}
where $s_i$ is the random service in bits offered by the system in time slot $i$, $C$ is the number of transmitted symbols per time slot, and $\gamma_i$ is the instantaneous SNR.
Then, we can obtain the cumulative service process in the SNR domain as
\begin{equation}
\label{eq:service_process_snr}
  \mathcal{S}(\tau,t)  = \prod_{i=\tau}^{t-1}  e^{s_i} = \prod_{i=\tau}^{t-1} g\left(\gamma_i\right) = \prod_{i=\tau}^{t-1} \left(1 + \gamma_i\right)^{\mathcal{C}},
\end{equation}
where $\mathcal{C} = C / \log 2$. 
Furthermore, in case of first-come first-served order, the delay at time $t$ is obtained
as follows
\begin{equation}
\label{eq:delay_snr}
  W(t)=\mathcal{W}(t)=\inf \{i \geq 0 : \A(0,t) / \D(0,t+i) \leq 1 \}.
\end{equation}
An upper bound $\varepsilon$ for the delay violation probability $\mathrm{Pr}\left[W(t)>w^\varepsilon\right]$ can be derived based on a transform of the cumulative arrival and service processes in the SNR domain using the moment bound.
In~\cite{alzubaidy_ton}, it was shown that such a violation probability bound for a given $\wepsilon$ can be obtained as
$\inf\limits_{s>0}\left\{\Kd(s, t+\wepsilon,t)\right\}$.

\noindent We refer to the function $\Kd\left(s,\tau,t\right)$ as the \textit{kernel} defined as
 \begin{equation}
\label{eq:function_M_Hussein}
  \Kd(s, \tau, t) = \sum_{i=0}^{\mathrm{min}(\tau,t)} \mathcal{M}_{\mathcal{A}}(1+s,i,t) \mathcal{M}_{\mathcal{S}}(1-s,i,\tau),
\end{equation}
where the function $\mathcal{M}_{\mathcal{X}}\left(s\right)$ is the Mellin transform \cite{Book:mellin} of a random process, defined as
\begin{equation}
\mathcal{M}_{\mathcal{X}}\left(s,\tau,t\right) = \mathcal{M}_{\mathcal{X}\left(\tau,t\right)} \left(s\right) = \EE\left[\mathcal{X}^{s-1}\left(\tau,t\right)\right],
\label{eq:Mellin_Definition}
\end{equation}
for any $s \in \mathbb{C}$, whenever the expectation exists. We restrict our derivations in this work to real valued\footnote{We note that by definition of $\mathcal{X}(\tau,t) = e^{X(\tau,t)}$, the Mellin transform $\mathcal{M}_{\mathcal{X}}\left(s,\tau,t\right) = \EE\left[e^{(s-1)X(\tau,t)}\right]$ after substitution of parameter $s = \theta+1$ implies also a solution for the moment-generating function (MGF), that is the basis of the effective capacity model \cite{wu} and of an MGF network calculus \cite{fidler_mgf}.} $s \in \mathbb{R}$. Introducing the Mellin transform in the performance analysis of wireless fading channels results into tractable mathematical expressions when computing network calculus bounds, which in turn results in scalable closed-form solutions. Using the assumption of stationary increments of the arrival and service processes, their Mellin transforms become independent of the time instance, and hence we write $\mathcal{M}_{\mathcal{X}}\left(s, t - \tau\right)$.
%
In addition, as we only consider stable queuing systems in steady-state, the kernel becomes independent of the time instance $t$ and we denote $\Kd\left(s,t+\wepsilon,t \right) \overset{t \to \infty}{=} \Kd\left(s,-\wepsilon\right)$.

The strength of the Mellin-transform-based approach becomes apparent when considering block-fading channels.
The Mellin transform for the cumulative service process in SNR domain is given by
\begin{equation}
\label{eq:mellin_transform_service_basic}
 \mathcal{M}_{\mathcal{S}}\left(s,\tau,t\right)=\prod_{i=\tau}^{t-1} \mathcal{M}_{g(\gamma)}\left(s\right)=\mathcal{M}_{g(\gamma)}^{t-\tau}\left(s\right)= \mathcal{M}_{\S}\left(s,t-\tau\right)\, , \nonumber
\end{equation}
where $\mathcal{M}_{g(\gamma)}\left(s\right)$ is the Mellin transform of the stationary and independent service increment $g\left(\gamma\right)$ in the SNR domain.
The function $g\left(\cdot\right)$ is associated with the channel capacity of a point-to-point fading channel as defined by Eq.~\eqref{eq:Shannon}. However, it can also model more complex system characteristics, most importantly scheduling effects. It is important here to note that, $g\left(\cdot\right)$ can represent \emph{any} service capacity, as long as it is being represented in the SNR domain. In this work we focus, however, on wireless fading channels. 

Assuming the cumulative arrival process in SNR domain to have stationary and independent increments we denote the corresponding Mellin transform by \mbox{$\M_{\A}\left(s,t - \tau\right) = \prod_{i=\tau}^{t-1}\M_{\alpha}(s) = {\M_{\alpha}}^{t-\tau}(s)$}. Substituting these two cumulative processes in Eq.~\eqref{eq:function_M_Hussein}, for the general form of the steady-state kernel for a communication channel we get 
\begin{equation}
  \Kd\left(s,-w\right) = \frac{\M_{g\left(\gamma\right)}^w\left(1 - s\right)}{1 - \M_{\alpha}\left(1 + s \right) \M_{g\left(\gamma\right)}\left(1 - s \right)}
  \label{eq:delay_kernel}
\end{equation}
for any $s > 0$, whenever the following stability condition holds,
\begin{equation}
\label{eq:stability_cond}
  \M_{\alpha}\left(1 + s \right) \M_{g\left(\gamma\right)}\left(1 - s \right) < 1.
\end{equation}

Assuming Rayleigh fading, i.e., an exponentially distributed SNR with average $\bar{\gamma}$ at the receiver, the Mellin transform of the service process results into \cite{alzubaidy_ton}
\begin{equation}
\label{eq:service_process_mellin}
  \mathcal{M}_{g(\gamma)}\left(s\right)=e^{\frac{1}{\bar{\gamma}}} {\bar{\gamma}}^{s-1} \Gamma\left(s,{\bar{\gamma}^{-1}}\right).
\end{equation}
\noindent where $\Gamma(x,y)=\int_y^{\infty} t^{x-1}e^{-t}dt$ is the incomplete Gamma function.
Substituting this in Eq.~\eqref{eq:delay_kernel}, the steady-state kernel for a Rayleigh-fading wireless channel is given by
\begin{equation}
  \Kd\left(s,-w\right) =
  \frac{\left(e^{\nicefrac{1}{\bar{\gamma}}} {\bar{\gamma}}^{-s } \Gamma\big(1-s,\frac{1}{\bar{\gamma}}\big)\right)^{w}}{1- \M_{\alpha}\left(1 + s\right)  e^{\nicefrac{1}{\bar{\gamma}}}  {\bar{\gamma}}^{-s } \Gamma\big(1- s ,\frac{1}{\bar{\gamma}}\big)} \, ,
  \label{eq:delay_kernel_rayleigh}
\end{equation}
for any $s > 0$ and under the stability condition in Eq.~(\ref{eq:stability_cond}). By optimizing Eq.~\eqref{eq:delay_kernel_rayleigh} over all $s>0$, i.e., $\inf_{s>0} \{\Kd(s,-w)\}$, a bound on the delay violation probability $\varepsilon$ for a given $w$ can be obtained.

An asymptotic lower bound, which coincides with the upper bound when $w \rightarrow \infty$, can be obtained using the large-deviation theory~\cite{alzubaidy_ton}.

\subsection {Recursive Formula}
\label{sec:recursion}

A further advantage of network calculus is the ability to capture a cascade of service processes into a  joint service curve using the server concatenation theory. This property is especially useful for performance analysis of multi-hop networks. Similar to the $(\min,+)$ algebra, the joint service curve is obtained through the end-to-end convolution according to the \mx~network calculus. For a path $\mathbb{L}$ it is defined as follows:
\begin{equation}
\S^{\mathbb{L}}(\tau,t)=\S_1 \otimes \S_2 \otimes ... \otimes \S_N(\tau,t),
\label{eq:S_convoluted}
\end{equation}
\noindent where $\S_1 \otimes \S_2(\tau,t)=\inf_{\tau \leq i \leq t} \{\S_1(\tau,i)\cdot \S_2(i,t)\}$.

A Mellin transform of the \mx~convolution is not available. Instead, we define a bound on the Mellin transform of the end-to-end service based on the server concatenation defined above and using the union bound, also known as Boole's inequality. 

Let $\mathcal{S}_1(\tau,t)$ and $\mathcal{S}_2(\tau,t)$ be two independent non-negative bivariate random processes representing the service processes of link 1 and 2, respectively. 
For $s<1$, the Mellin transform of the $(\min,\times)$ convolution of ${\cal{S}}_1$ and ${\cal{S}}_2$, denoted by $\mathcal{S}_1\otimes\mathcal{S}_2(\tau,t)$, is bounded by
\begin{equation}
\label{eq:convolution_mellin}
  \mathcal{M}_{\mathcal{S}_1\otimes\mathcal{S}_2}(s,\tau,t)\leq \sum_{i=\tau}^t \mathcal{M}_{\mathcal{S}_1}(s,\tau,i) \cdot \mathcal{M}_{\mathcal{S}_2}(s,i,t) \nonumber
\end{equation}

Hence, the corresponding Mellin transform of the path $\mathbb{L}$ can be bounded by~\cite{alzubaidy_ton}:
\begin{align}
\label{eq:mellin_transform_path}
\M_{\S^{\mathbb{L}}}(s,\tau,t) &\leq \sum_{i_1=i_0}^{t} \sum_{i_2=i_1}^{t}\dots \!\!\! \sum_{i_{N-1}=i_{N-2}}^{t}  \!\!\! \M_{\S_1}\left(i_1 - i_0\right) \cdot 
\M_{\S_2}\left(i_2 \! - \! i_1\right) \dots \M_{\S_N}\left(i_N \! - \! i_{N-1}\right)\!\! \! \notag \\
&  = \sum_{i_1 \dots i_{N-1}}^{t}  \prod_{j=1}^{N} \M_{{\cal{S}}_j}^{i_j - i_{j-1}}\left(s\right), 
\end{align} 
\noindent with $\tau = i_0 \leq i_1 \leq \dots \leq i_N = t$.
Notice that $\M_{{\cal{S}}_j}\left(s\right)$ denotes the Mellin transform of the (stationary) SNR service increments of link $j$.

As one may notice from Eq.~\eqref{eq:convolution_mellin}, this results in a cumbersome computation, especially for links having different channel gain distribution, since $N$ convolution processes have to be computed, each of them depending on $t$. Hence, we now present a significant simplification of Eq.~\eqref{eq:convolution_mellin} into a mathematically easier-to-grasp analytical solution, avoiding the tedious task of performing $N$ nested sums in Eq.~\eqref{eq:mellin_transform_path}. To this end, we define $\Kd^{\mathbb{L}}$ as the kernel for a path $\mathbb{L}$ containing $N$ links, similar to Eq.~\eqref{eq:function_M_Hussein}, where we replace $\mathcal{M}_{\mathcal{S}}(1-s,i,\tau)$ with $\mathcal{M}_{\mathcal{S}^{\mathbb{L}}}(1-s,i,\tau)$ for $\S^{\mathbb{L}}$ defined in Eq.~\eqref{eq:S_convoluted}.
Once that Mellin transform can be determined, a probabilistic end-to-end delay bound $w^{\mathbb{L}}(\varepsilon)$ for path $\mathbb{L}$ can be computed using Eq.~(\ref{eq:function_M_Hussein}) as the smallest $w$ that satisfies the following inequality \cite{alzubaidy_ton}

\begin{align}\label{eq:DelayBound}
\inf_{s \ge 0}\left \{  \Kd^{\mathbb{L}} (s, -w^\varepsilon)  \right \} \le \varepsilon \, . 
\end{align}

Let $m$ and $N$ refer to the $m^{\mathrm{th}}$ and $N^{\mathrm{th}}$ link of path $\mathbb{L}$, respectively.

\begin{theorem} \label{thm:1}
Given a path $\mathbb{L} \setminus \{N\}$ of links with independent and heterogeneously distributed service processes,  with kernel $\Kd^{\mathbb{L} \setminus \{N\}}$; then $\Kd^{\mathbb{L}}$ can be obtained in terms of $\Kd^{\mathbb{L} \setminus \{N\}}$  as follows
{
\begin{equation}
\begin{aligned}
 \Kd^{\mathbb{L}} \left(s,-w\right)& =  \frac{\M_{g(\gamma_N)}\left(1-s\right)}{\M_{g(\gamma_N)}\left(1-s\right) - \M_{g(\gamma_{m})}\left(1-s\right)} \Kd^{\mathbb{L} \setminus \{m\}} \left(s,-w \right) \notag \\
&+ \frac{\M_{g(\gamma_{m})}\left(1-s\right)}{\M_{g(\gamma_{m})}\left(1-s\right) - \M_{g(\gamma_{N})}\left(1-s\right)} \Kd^{\mathbb{L} \setminus \{N\}} \left(s,-w\right) 
\label{eq:M_hop_delay_kernel}
\end{aligned}
\end{equation}   }
for any $m \in \{1,2,\dots ,N-1\}$.
\end{theorem}

\noindent The proof for  Theorem \ref{thm:1} is given in Appendix~\ref{app:theorem_proof}. Note that this result has been previously presented in~\cite{icc15}. 

A direct consequence of Theorem~\ref{thm:1} and Eq.~\eqref{eq:DelayBound} is that the delay bound for path $\mathbb{L}$ can be obtained from recursively computing the kernel according to the theorem. In this recursion, the number of summands increases with the number of hops. For an $N$-hop path there are $2^{N-1}$ summands, as each geometric sum results into two summands. Furthermore, the stability condition in Eq.~\eqref{eq:stability_cond} needs to hold for every individual link $n \in \{1,..,N\}$, i.e.,:
 \begin{equation}
   \max_{n} \left( \M_{\alpha}\left(1 + s \right) \cdot \M_{g\left(\gamma_n\right)}\left(1 - s \right)  \right) < 1 \nonumber
 \end{equation} 
 \noindent has to be fulfilled. 
Finally, note in particular that in principle Theorem~\ref{thm:1} can be generalized to any link $n$ and the path $\mathbb{L} \setminus \{n\}$. This allows an efficient recomputation of the kernel in case that any of the links of the path change their primary distribution, for instance due to a changed propagation environment. In contrast, in case of using Eq.~\eqref{eq:mellin_transform_path} if a single link changes its distribution, a complete recomputation of the joint service curve characterization has to be performed, which is significantly more complex. 


We show next that the kernel described by Theorem~\ref{thm:1} is convex in $s > 0$. The following theorem states this convexity. 
\begin{theorem} \label{thm:2}
The steady-state kernel $\K(s,-w)$ for a communication channel,
\begin{equation}
  \Kd\left(s,-w\right) = \frac{\M_{g\left(\gamma\right)}^w\left(1 - s\right)}{1 - \M_{\alpha}\left(1 + s \right) \M_{g\left(\gamma\right)}\left(1 - s \right)}, \nonumber
\end{equation}
is convex in $s, \forall s>0$ for which the stability condition $\M_{\alpha}\left(1 + s \right) \M_{g\left(\gamma\right)}\left(1 - s \right) < 1$ 
holds. Furthermore, the end-to-end kernel $\K^{\mathbb{L}}(s,-w)$ for a multi-hop path $\mathbb{L}$ is convex in $s$, for every $s$ within the stability interval.
\end{theorem} 
The proof of Theorem~\ref{thm:2} is given in Appendix~\ref{app:convexity}.

The convexity of the kernel confirms the existence of a unique optimum and motivates us to specify bound-based algorithms for optimization, i.e., transmit power allocation along the path which reaches the optimum. This is presented in Section~\ref{sec:algorithm}. 

We note here that, both Theorem~\ref{thm:1} and Theorem~\ref{thm:2} hold for \emph{any} kernel for a communication channel, whose service is defined in the SNR domain. In this paper, however, we focus our analysis on wireless fading channels and two kernel types - namely the ones specified for a Shannon- and WirelessHART-based service. 

{
\subsection{Heterogeneous Path with Cross-Traffic}

The proof of Theorem~\ref{thm:1}, given in the appendix, shows that the recursion obtained for the end-to-end delay bound results from the recursion in the Mellin transform of the joint service curve. 
As a result, the stepwise construction of a multi-hop path's service curve will further simplify the computation of other elements in \mx~network calculus, such as the delay bound. 
We demonstrate this by providing the leftover service curve for the considered path when cross-traffic is present, i.e., additional flows other than $A(\tau,t)$ are sharing the intermediate links in the path that the through-flow traverses.
Let the SNR arrival processes of the cross-traffic at each intermediate link be i.i.d. and denoted by $\mathcal{A}_c(\tau,t)=e^{k_c(t-\tau)}$, where we assume a constant arrival rate of $k_c$ bits per time slot. 
Assume further that the arrivals from the original through-flow and the cross-traffic as well as the service processes at each link are independent. 
We can then compute a bound on the Mellin transform of the end-to-end service process offered to the through-flow, i.e., the \emph{leftover service curve} using the following result:

\begin{theorem} \label{lem_2}
Consider a flow traversing a cascade of wireless fading channels. The service at each node is shared by the through-flow and an independent cross-flow characterized by the SNR arrival process $\mathcal{A}_c (\tau,t)$. Let ${\mathcal{S}_c}^\mathbb{L}(\tau,t)$ denote the end-to-end leftover service provided to the through-flow. Then, $\forall s<1$, 
\begin{equation}
\label{eq:left-over-sc}
\begin{aligned}
 \M_{\mathcal{S}_c}^\mathbb{L}(s,\tau,t) \leq & \left(\frac{\M_{g(\gamma_N)}(s)}{\M_{g(\gamma_N)}(s) - \M_{g(\gamma_m)}(s)} \cdot \M_{\mathcal{S}_c}^{\mathbb{L}\setminus \{m\}} (s,\tau,t) \right)\\
&  +\left(\frac{\M_{g(\gamma_m)}(s)}{\M_{g(\gamma_m)}(s) - \M_{g(\gamma_N)}(s)} \cdot \M_{\mathcal{S}_c}^{\mathbb{L}\setminus \{N\}} (s,\tau,t)\right),
\end{aligned}
\end{equation}

\noindent for any $m \in \{1,2,\dots ,N-1\}$ and $\M_{\mathcal{S}_c}^{\{1\}}(s,\tau,t)=\left(e^{k_c(1-s)}\cdot \M_{g(\gamma_1)}(s)\right)^{t-\tau}$.
\end{theorem} 

\begin{proof}
 According to \emph{Lemma 1} in~\cite{alzubaidy_ton} we obtain the Mellin transform of the leftover service curve for a single channel:
\begin{equation}
\label{eq:single_channel_left_over_sc}
\begin{aligned}
\M_{\calSc}^{\{1\}}(s,\tau,t) & =\M_{\nicefrac{\calS}{{\cal{A}}_c}}(s,\tau,t)\\
& = \M_{g(\gamma_1)}(s,\tau,t)\cdot \M_{{\cal{A}}_c}(2-s,\tau,t) \\
& \leq \left(e^{k_c(1-s)}\cdot \M_{g(\gamma_1)}(s)\right)^{t-\tau},
\end{aligned}
\end{equation}
since the Mellin transform of a quotient of two independent random variables $X$ and $Y$ is given by $\M_{\nicefrac{X}{Y}}(s)=\EE[X^{s-1}]\EE[Y^{1-s}]=\M_{X}(s)\cdot\M_{Y}(2-s)$.

Substituting $e^{k_c(1-s)}\M_{g(\gamma_N)}(s)$ for $\M_{g(\gamma_N)}$ in the joint service curve derivation given in Appendix~\ref{app:theorem_proof}, the expression given by Eq.~\eqref{eq:left-over-sc} follows by applying the recursion to the leftover service curve of path $\mathbb{L}$, $\M_{\calSc}^{\mathbb{L}}(s,\tau,t)$.
\end{proof}
 }

\section{Bound-Based Power Minimization Algorithms}
\label{sec:algorithm}
As already mentioned, an essential aspect of industrial wireless networks, besides the importance of QoS-awareness, is their energy-efficient operation. 
Especially for battery-powered network devices, often attached to machines in order to control them or measure their functional status, it is important to prolong network partitioning time by extending nodes' battery lifetime. 
Since the radio chip is usually one of the largest consumers of energy in low-power networks~\cite{khader_willig}, one way of providing energy efficiency is to minimize the transmit power, as one of the easily modifiable parameters in wireless transceivers.
In addition, minimizing transmit power not only increases energy savings, but also reduces potential interference to neighboring transmissions. 
In a multi-hop setting, power optimization mainly needs to take two issues into account.
On the one hand, heterogeneous link statistics can be exploited to reduce power consumption.
On the other hand, heterogeneous battery states (i.e., energy levels) affect the transmit power setting.
Thus, in this section we develop and present algorithms that take these effects into account in order to minimize transmit power or maximize network lifetime under statistical end-to-end constraints as represented by the above presented bound. 

\subsection{Transmit Power Minimization Algorithm}
\label{subsec:power_min_alg}

We initially raise the following question: What is the optimal average SNR, i.e., minimal sum transmit power needed on all links along a path to meet a target end-to-end delay $w^{\varepsilon}$ with probability $1-\varepsilon$? 
This question is difficult to answer, as for the violation probability $\varepsilon$ no accurate analytical model exists (to date) that can relate it precisely to the average SNRs.
The only option is to resort to system simulations to determine the corresponding SNRs.
In contrast, we propose to base system optimization on the multi-hop delay bound, as presented above. 
Hence, we are interested in the solution of the following optimization problem for a given multi-hop path $\mathbb{L}=\{1,..,N\}$:
\begin{equation}
\begin{aligned}
&\min &  &\sum_{j=1}^N {\Ptx}_j, \nonumber \\ 
&\text{s. t.}
 &  &\inf_{s>0}  \{ \K^{\mathbb{L}}(s,-w^{\varepsilon})\} \leq \varepsilon \,  \nonumber \\ 
  &  & & \max_{j} \left( \M_{\alpha}\left(1 + s \right) \cdot \M_{g\left(\gamma_j\right)}\left(1 - s \right)  \right) < 1 \label{eq:st_cond_multi}. 
\end{aligned}
\end{equation}
Due to the complexity of the kernel function and the stability condition (see Eq.~\eqref{eq:delay_kernel_rayleigh}), no analytical solution for the bound-based optimal SNRs can be derived. 
Instead, we propose a binary search algorithm in two dimensions (along $s>0$ and along the SNRs) to solve the given minimization problem (see Algorithm~\ref{fun:power_min_short} in Appendix~\ref{app:pseudo_codes}). 
Notice that minimizing the SNRs $\gammabar_j$ leads to minimization of the transmit power per hop, since $\gammabar_j=\nicefrac{p_j}{\sigma^2}$. 

As already stated in Theorem~\ref{thm:2}, the kernel described by Theorem~\ref{thm:1} is convex in $s$ for $s \in (0,b)$, where $b$ is the last point for which the stability condition in Eq.~\eqref{eq:stability_cond} holds. From this it follows that the proposed algorithm results in a global bound-based minimum.
Figure~\ref{fig:convexity_diff_SNR} illustrates the kernel of several links with different average SNR, where the instantaneous channel capacity is given by Eq.~\eqref{eq:Shannon}. The figure shows that the delay bound function of a single link $\K(s,-w)$ is convex in $s$ and monotone in $\gammabar$. For the computation of the kernel in the figure, we assumed a block-fading wireless link with constant arrival rate and random service increments that are characterized by the Shannon capacity.
\begin{figure}
\centering
  \includegraphics[scale=0.55]{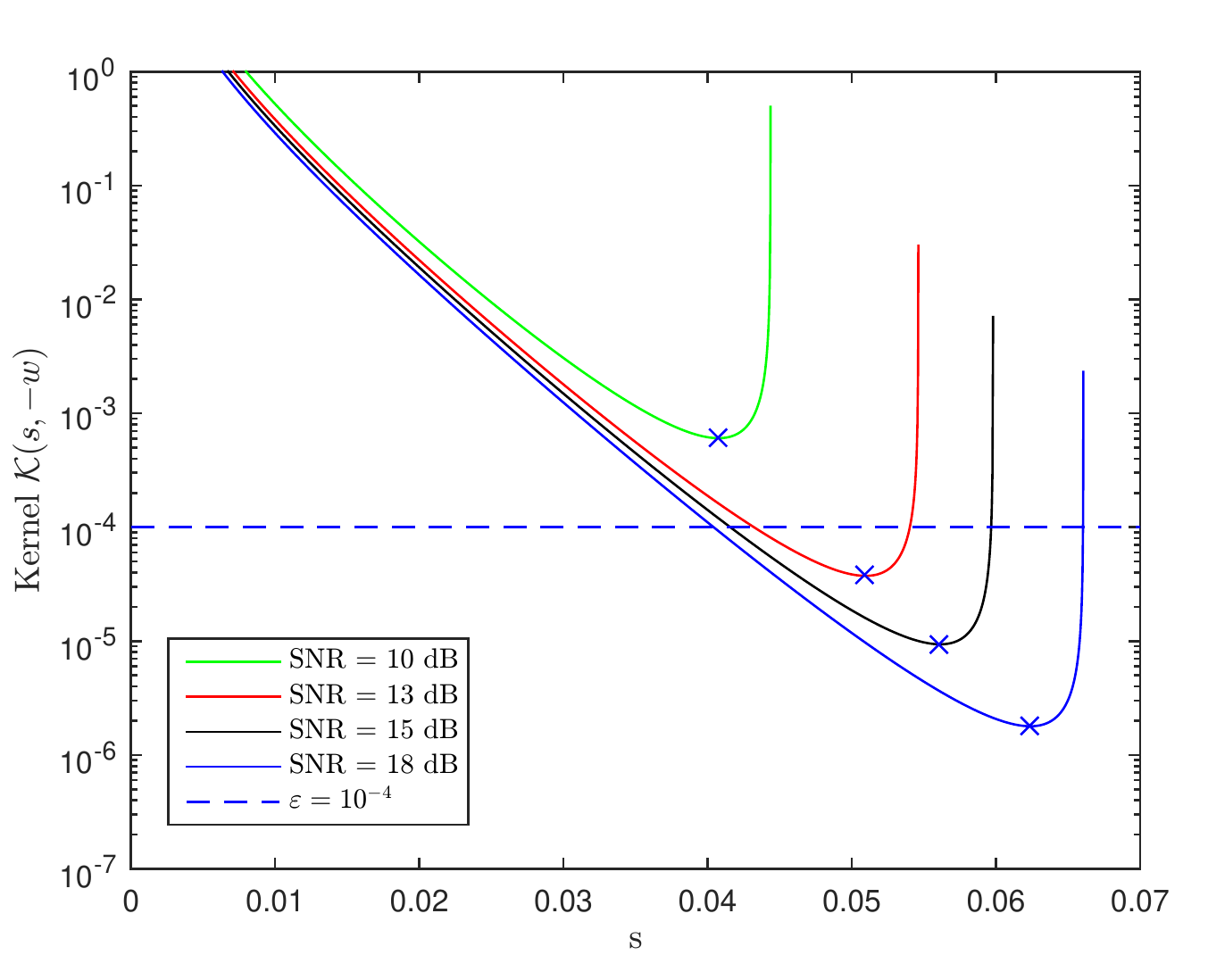}
\caption{The delay bound function $\K(s,-w)$ is convex in $s$. It is obtained for $r_a=50$ bits per time slot and target delay $w=5$ time slots. Its minimum is marked with a cross and shifts to the right as the average SNR on the link is increased.}
\label{fig:convexity_diff_SNR}
\end{figure}
We further notice that, as the SNR either increases or decreases, the optimal $s^{*}$ (which minimizes the delay bound function) moves to the right or to the left, repectively. 


For any given transmit power vector $\vPtx = \left\{p_1, \ldots p_n\right\}$ and resulting fixed SNR vector $\vgammabar = \left\{{\bar{\gamma}}_1, \ldots {\bar{\gamma}}_n\right\} $, the value $s^{*}$ for which $\K^{\L}(s,-w)$ is minimal, is determined by performing a binary search along the interval $(0,b)$.
The main idea here is to cut the interval $(0,b)$ into four areas through fixing five points (see Fig.~\ref{fig:show_s_points}), where $s_\mathrm{m}$ is the middle point of $(0,b)$. 
Based on this partition, the algorithm traces the gradients and splits the range where the minimum of $\K^{\L}(s,-w)$ is located.
The function is called recursively until the smallest size of an interval has been reached, defined with the input parameter $\Deltamin$. 
At this point, the middle point $s_\mathrm{m}$ of the last considered partition is returned as $s^{*}$, i.e., as the point $s$ for which $\K^{\L}(s,-w)$ reaches its minimum.

\begin{figure}
\centering
  \includegraphics[scale=0.6]{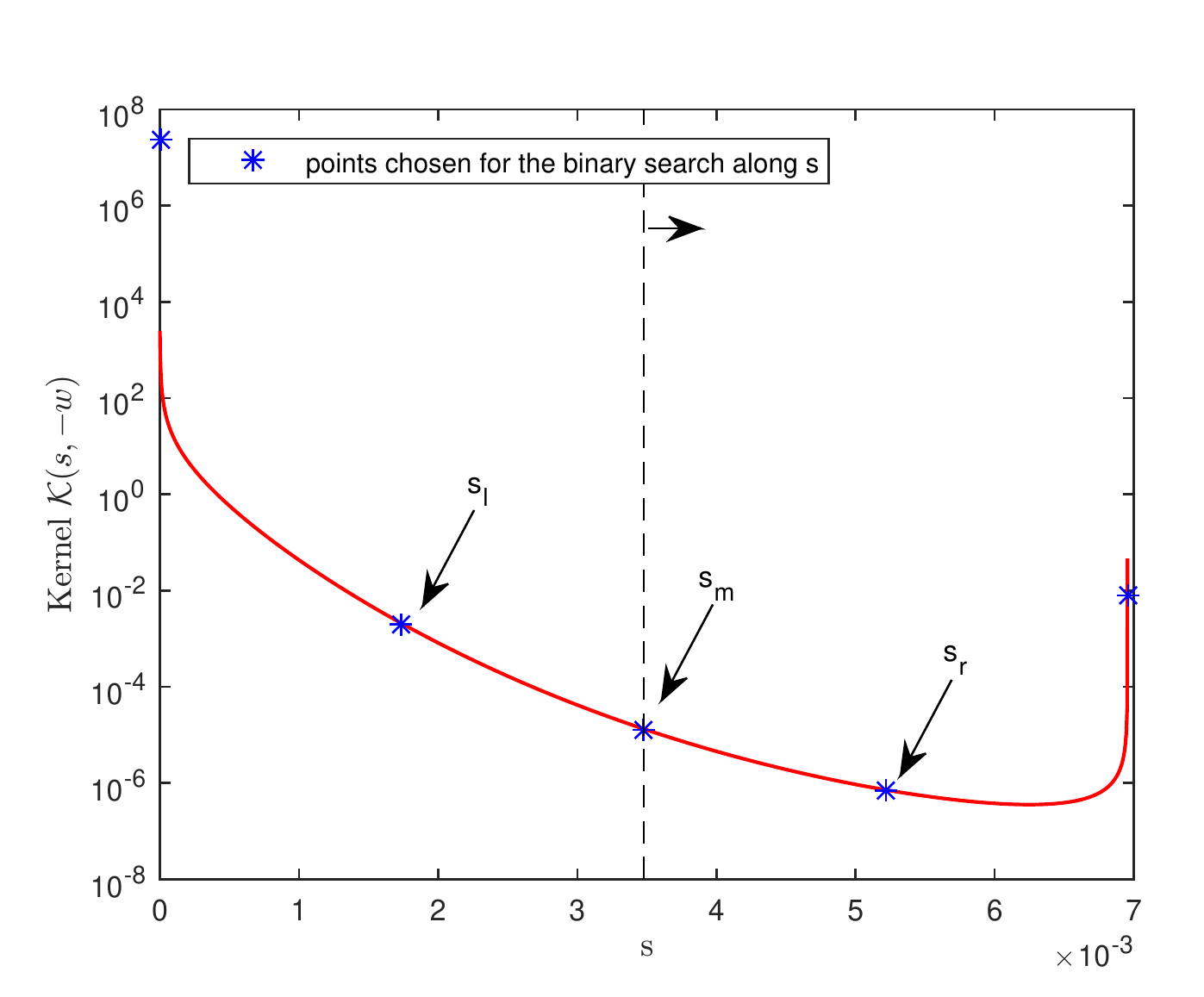}
\caption{The chosen points within the interval $(0,b)$ for the binary search along dimension $s$.  $s_\mathrm{l}$ and $s_\mathrm{r}$ mark the middle of the left and right half of the interval, respectively, while $s_{\mathrm{start}}$ and $s_{\mathrm{end}}$ represent 0 and $b$, respectively.}
\label{fig:show_s_points}
\end{figure}

For the search in the second dimension along the $\gammabar$ dimension (see Algorithm~\ref{fun:search_s} in Appendix~\ref{app:pseudo_codes}), we start by allocating a predefined maximal transmit power $p_{\mathrm{max}}$ to each node along the path. 
In each iteration, the gradient of the end-to-end kernel $\K^{\L}(s^{*}, -w)$ is computed for every link on the path, i.e., \mbox{$\forall n \in \{1,..,N\}$}. The smallest gradient defines the link $n$ whose transmit power is going to be changed in that iteration.
$\vgammabar_n$ is the vector of average SNRs per link in which the $n$-th SNR is replaced by the SNR obtained when $p_{n}$ is decreased for some predefined $\deltaPtx$ (line~\ref{line:gradient} in Algorithm~\ref{fun:power_min_short} Appendix~\ref{app:pseudo_codes}). 
In each iteration a new kernel is computed (denoted with $\hat{\varepsilon}$). 
In case the newly computed kernel is bigger than the target violation probability $\varepsilon$, $\deltaPtx$ is halved, so that a decrease of the transmit power is further possible.
 $p_n$ is halved until it reaches a predefined minimal value $\deltaPtxmin$. 
 The algorithm returns the current vector $\vPtx$ as an optimal one, either when all links are assigned with the minimal possible transmit power or the smallest possible $\deltaPtx$ has been reached and the transmit power along the links cannot be further reduced. 
 It may happen though, that the target delay cannot be met. 
 In the optimal case the algorithm exits when the obtained kernel has approached the target violation probability from below, i.e., $\hat{\varepsilon} \in (\varepsilon-\deltaEpsilon,\varepsilon)$ for some predefined $\deltaEpsilon$. 

The obtained bound-based solution for $\vgammabar$ and $\vPtx$ is quasi-optimal, since the binary-search algorithm approaches to the optimal $s^{*}$ and the target $\varepsilon$. 
Nevertheless, the input parameters $\deltaEpsilon, \deltaPtxmin$ and $\Ptxmax$ can be used to make a trade-off between the algorithm's precision and its performance, i.e., the needed number of iterations to reach an optimal solution.


\subsection{Network Lifetime Maximization Algorithm}
\label{subsec:lifetime_max_alg}

For industrial automation applications, reducing the transmit power results into lower interference with neighbouring networks, which leads  to  better coexistence of multiple wireless technologies within the same area. Furthermore, it  increases the energy-efficiency of the wireless network, which is crucial for applications where battery-powered devices are used. Early battery exhaustion will cause a shorter overall network lifetime as well as a potential premature network partitioning.  Therefore, extending battery lifetime is another important aspect in the performance analysis of industrial wireless applications. 

In applications where network lifetime is more important than pure energy saving, the bound-based power-minimization algorithm defined in the previous section may not be ideal. It is worth noting that in the case of heterogeneous multi-hop networks, the proposed algorithm may result in an unequal depletion of the battery energy levels at different hops, which may result in shorter network lifetime. An alternative approach is to allow nodes with lower energy levels to use lower transmit power, which may result in an increased delay, while other nodes pick-up the slack by increasing their transmit power to compensate for the extra delay introduced by that node.
We therefore propose a separate algorithm for lifetime maximization, which is a modification of the bound-based power-minimization algorithm defined in the previous section, which aims at maximizing the network lifetime. 
In this work, we assume a network (or a path) is no longer useful when any of the nodes' battery is fully depleted. Nevertheless, the algorithm can also be used, with few modifications, to handle more resilient wireless networks design. 
The difference from the previous algorithm is mainly related to the decision regarding whose link's transmit power to decrease in each subsequent iteration of the algorithm. In the new algorithm, it is redefined into choosing the transmitter that has the least charged battery at that moment of time. For this purpose we look at the battery full states (denoted by the vector $\vec{\mathrm{B}}$) of the nodes along the path $\L$. The goal is to maximize the minimal battery lifetime (or duration of battery operation) $\theta_j$ among all batteries. Each relay node can be in one of the following states: idle, send and receive. The battery consumption during the idle and the receive phase is dependent on the transceiver, while the energy consumption in the sending state is mainly dictated by the transmit power. Notice further, that the source node cannot be in receive mode, while the destination does not send packets and is therefore excluded from the decision process. Having a time slotted system with slot length $T$, the rest of the time slot assigned to a node, in which it neither sends nor receives the packet, is spent in an idle state. 

To handle the effect described above, we formulate the following bound-based network lifetime maximization problem:  
\begin{equation}
\begin{aligned}
& &  &\max\min_{j=1 \dots N}\{\theta_j\}, \theta_j=\frac{B_j}{p_j \cdot T} \nonumber \\ 
&\text{s.t. }
 &  &\inf_{s>0}  \{ \K^{\L}(s,-w^{\varepsilon})\} \leq \varepsilon \,  \nonumber \\ 
  &  & & \max_{j} \left( \M_{\alpha}\left(1 + s \right) \cdot \M_{g\left(\gamma_j\right)}\left(1 - s \right)  \right) < 1. 
\end{aligned}
\end{equation}
\noindent In comparison to the transmit power minimization algorithm presented in Sec.~\ref{subsec:power_min_alg}, the network lifetime maximization algorithm selects the node with the minimal battery duration (the vector of battery durations is denoted by $\vec{\mathrm{\theta}}$) as a candidate whose transmit power will be reduced in that iteration. The assigned transmit power can be selected in the interval $(\Ptxmin,\Ptxmax)$, both depending on the chosen hardware. Similar to the gradient-based algorithm, the transmit power is reduced in steps of $\deltaPtx$ until the resulting delay bound function (computed using Theorem~\ref{thm:1}) is bigger than the target one. Each time this is the case, $\deltaPtx$ is halved until some predefined $\deltaPtxmin$ has been reached. The pseudo-code of the network lifetime algorithm is given in Algorithm~\ref{fun:battery_alg} in Appendix~\ref{app:pseudo_codes}. Among the above mentioned parameters, the QoS requirements $w$ and $\varepsilon$, together with the payload size $r_a$ and the maximal frame size $k_a$, that can be sent per time slot on the channel, are input parameters to the algorithm.  


\section{Numerical Evaluation}
\label{sec:num_res}

In this section, we present numerical evaluations of the power minimization algorithms based on the end-to-end delay bound over heterogeneous links. 
In the following subsections we then focus on the power minimization algorithms. 
In Sec.~\ref{subsec:power-min} we evaluate our suggested algorithm for various path compositions. 
In Sec.~\ref{subsec:lifetime-max} we present results that correspond to network's lifetime maximization, considering a more realistic transceiver node model in a WirelessHART network based on the IEEE 802.15.4 standard~\cite{whart_online}. 
All presented results rely on Theorem~\ref{thm:1}, which was validated via simulations. In \cite{icc15} we show that the provided closed-form solution for the end-to-end delay bound is indeed an upper bound of the simulated delay violation probability for various multi-hop scenarios, observing a gap of approximately one order of magnitude. It is important to note that, the decay rate of the computed delay violation probability is exactly the same as the one obtained via simulations, which suggests that the bound is asymptotically tight. We refer the interested reader to \cite{icc15} for a detailed description of the validation. 


\subsection{Evaluation of the Power-Minimization Algorithm}
\label{subsec:power-min}
We now turn to the evaluation of the bound-based transmit power minimization algorithm presented in Section~\ref{subsec:power_min_alg}.  
Recall that the algorithm minimizes the total transmit power over all links of a multi-hop path based on the analytically determined end-to-end kernel according to Theorem~\ref{thm:1}. 
The target end-to-end delay violation probability $\varepsilon$ and delay $w$ are the QoS parameters passed to the algorithm.

\subsubsection{Methodology}
\label{subsubsec:power_min_methodology}
Through analytical evaluations, we benchmark the minimum total transmit power algorithm for various different scenarios, characterized by different path compositions.
We consider in general Rayleigh-fading links with different mean SNRs ${\bar{\gamma}}_{n}$. 
Also, we consider the Shannon capacity model to map a link's SNR to its service capacity.
The arrival flow in this investigation is fixed to $r_a=20$ bits per time slot.
We express link heterogeneity for a path $\mathbb{L}$ consisting of $N$ links using the norm of the vector $\mathbf l = (l_1, \dots, l_N)$, denoted by $R^{\mathbb{L}}$ and given by:
\begin{equation}
\label{eq:path_norm}
R^{\mathbb{L}}=\sum_{n=1}^N \sum_{m=n+1}^N | l_n - l_m | \; \mathrm{,}
\end{equation}
where $l_n$ denotes the length of link $n$, which reflects the path loss of the corresponding link and hence its service. 
Obviously, higher norm reflects higher link heterogeneity and vice versa.
In the following, we consider 3-hop $(N=3)$ paths with various node placements between a source and a destination located 60 m apart. 
Table~\ref{tab:paths} shows the exact scenarios (from almost homogeneous to strongly heterogeneous in ascending order) and their respective path norms used in the evaluations.  These scenarios are deliberately chosen to highlight the effect of link heterogeneity and relative distances between intermediate nodes on network performance and the power gain obtained using the proposed power minimization algorithm compared to a naive power allocation.  
\begin{table}[h]
\centering
\caption{Considered Path Compositions}
\label{tab:paths}
\begin{tabular}{ c  c }
  Link lengths in [m] & Path norm $R$ \\ \hline
  $[20, 19, 21]$ & 4 \\ 
  $[20, 30, 10]$ & 40 \\
  $[5, 28, 27]$ & 46 \\
	$[20, 35, 5]$ & 60 \\
	$[5, 40, 15]$ & 70 \\
	$[5, 50.5, 4.5]$ & 92 \\
\end{tabular}
\end{table}
Note that in the following we refer to the link with the longest distance as the \emph{critical link} (the link characterized with the highest path-loss). 

In order to evaluate the efficiency of our algorithm, we are in particular interested in the total power reduction it can achieve in comparison to other approaches.
For this, we consider two different comparison schemes that allocate a \textit{uniform} power value to all links:
\begin{itemize}
\item \emph{QoS-agnostic}: Each node along the path is assigned the same transmit power without considering QoS. In the numerical evaluation we use for this value the maximum available transmit power value of an IEEE 802.15.4 low-power transceiver~\cite{atmel}, which equals $\Ptxmax=4$ dBm. 
\item \emph{QoS-aware}: In this scheme, the transmit power is iteratively reduced equally for all nodes - starting from the maximal transmit power $\Ptxmax$ - until the obtained delay violation probability is larger than the target one ($\varepsilon$). Hence, as in the previous case with the QoS-agnostic scheme every node is assigned the same transmit power, however, the allocation is typically lower than $\Ptxmax$.
\end{itemize}
For all considered scenarios, we compute the minimum total transmit power as obtained from our algorithm, and compute afterwards the saving ratio or the power gain (in percent) that can be obtained in comparison to the QoS-agnostic scheme or the QoS-aware scheme.
A saving of $50 \%$ hence indicates that through our power minimization algorithm the total transmit power is half of the value resulting from the comparison scheme. 

\subsubsection{Numerical Results}
In Fig.~\ref{fig:shannon_sum_energy} we present the absolute required total transmit power in [mW] of our proposed algorithm for the discussed path scenarios over an increasing target delay when fixing the target delay violation probability.
As the target delay is increased, the required total transmit power decreases.
In addition, note that the total transmit power is higher for higher link heterogeneity.
This is due to the critical link which dominates the total transmit power consumption on the path and for which the delay can only be compensated for by other links up to a certain point.
Note that with a maximum transmit power of $4~\mathrm{dBm}$ per node, the total transmit power along the path equals to $7.5357~\mathrm{mW}$.
\begin{figure}
\centering
  \includegraphics[scale=0.6]{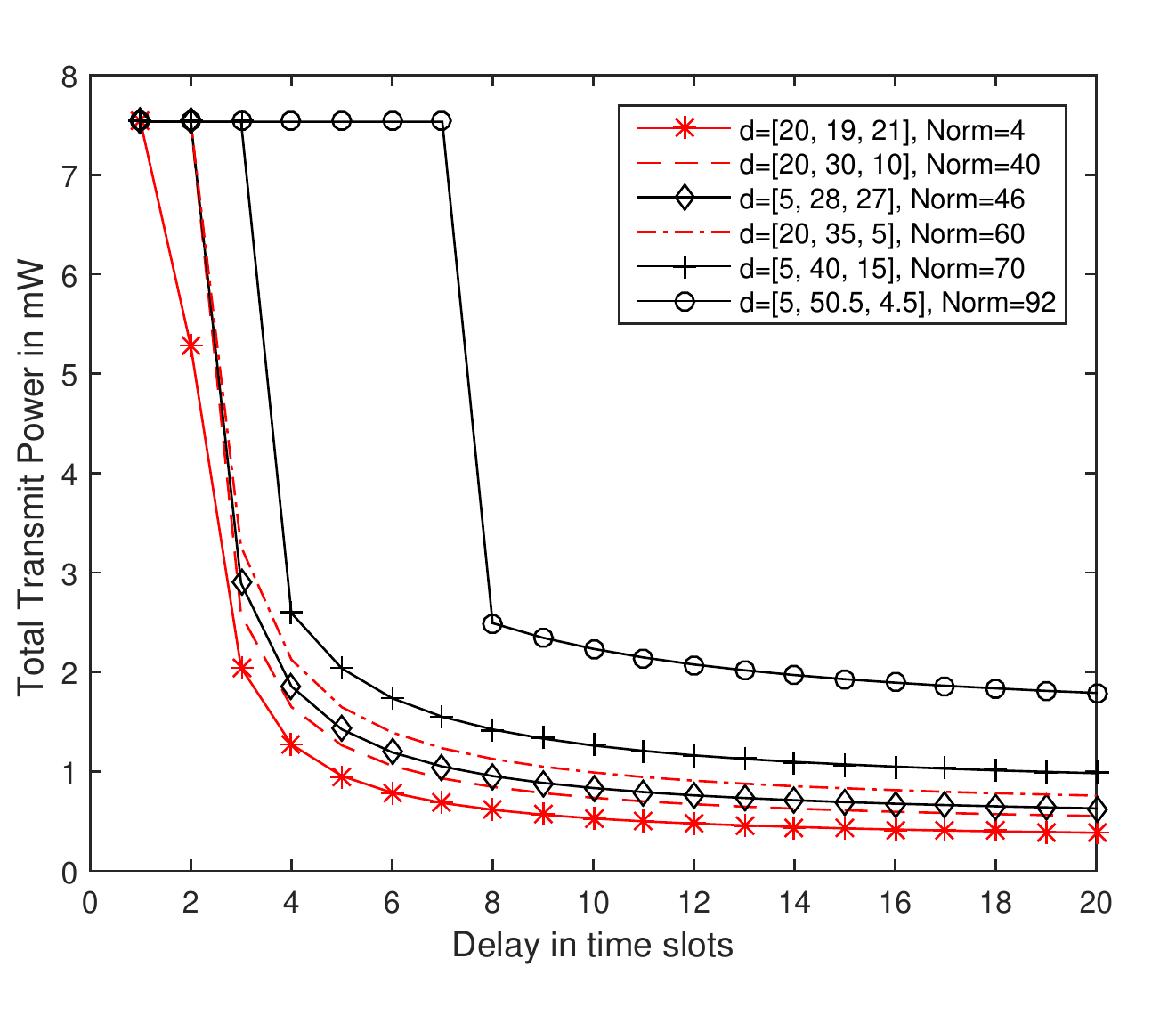}
\caption{Sum of the total transmit power for different 3-hop paths resulting from the power-minimization algorithm over an increasing delay target. The target delay violation probability is fixed at $\varepsilon=10^{-3}$.}
\label{fig:shannon_sum_energy}
\end{figure}

We next present the saving gains - in Fig.~\ref{fig:shannon_4dbm} in comparison to the QoS-agnostic scheme and in Fig.~\ref{fig:shannon_qos} in comparison to the QoS-aware scheme. 
For both figures we consider the same path compositions as above and vary the target delay while keeping the target delay violation probability fixed at $\varepsilon=10^{-3}$. 
\begin{figure}
\centering
  \includegraphics[scale=0.6]{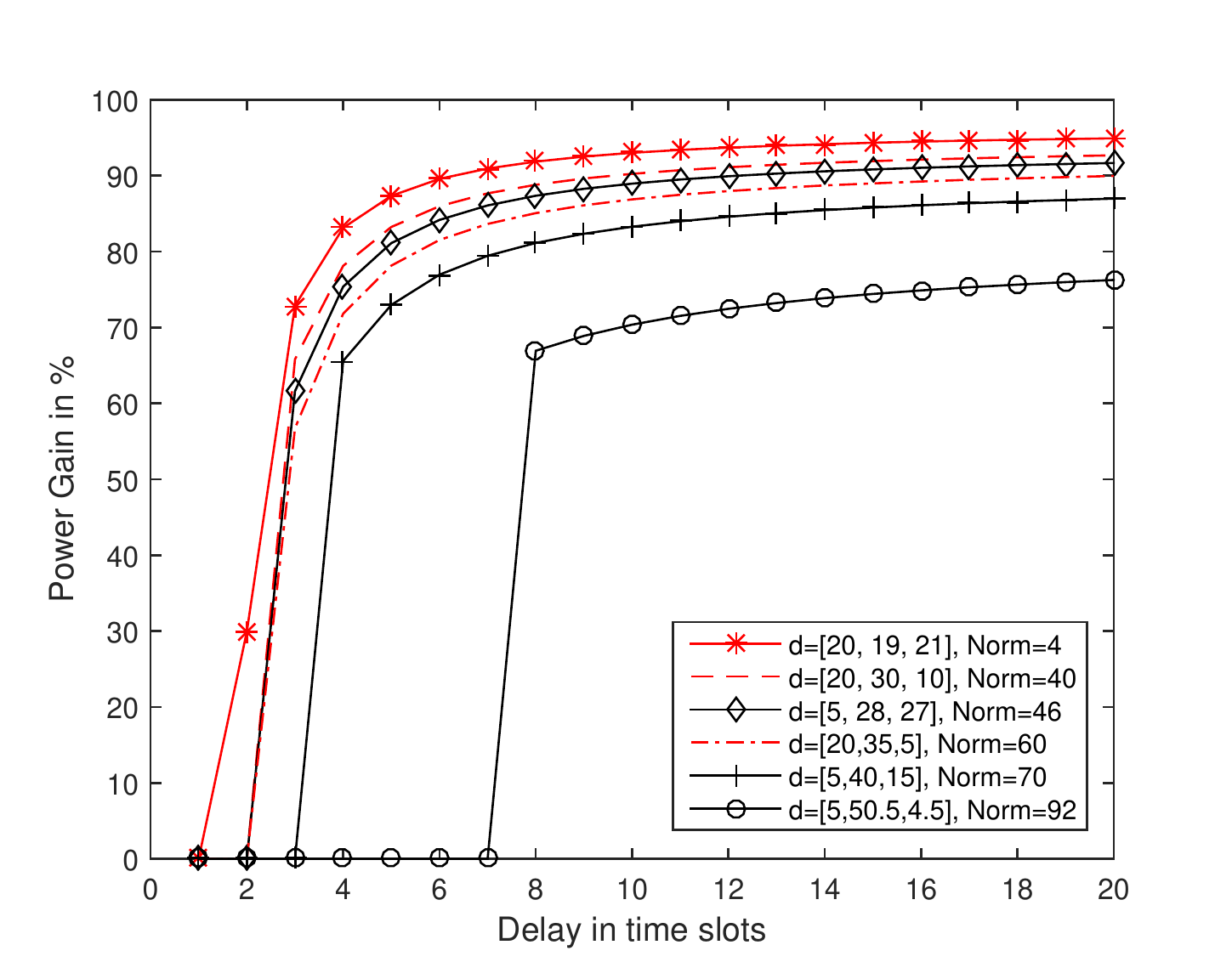}
\caption{Saving gain of the proposed power minimization algorithm in comparison to the QoS-agnostic scheme for an increasing target delay for various 3-hop path compositions. The target delay violation probability is fixed at $\varepsilon=10^{-3}$.}
\label{fig:shannon_4dbm}
\end{figure}
In Fig.~\ref{fig:shannon_4dbm} we observe initially that all saving gains increase for an increasing target delay.
This is a direct consequence from Fig.~\ref{fig:shannon_sum_energy}, as those values are compared to a fixed value of $7.5357~\mathrm{mW}$ in order to determine the saving gain.
Hence, it is also not surprising that the saving gain increases for more homogenous paths. 
In absolute terms, the saving gains are in the range of $70 \%$ to $90\%$ which nevertheless shows the potential of the proposed algorithm.
If we switch over to the saving gains in comparison to the QoS-aware scheme different observations can be made (see Fig.~\ref{fig:shannon_qos}).
\begin{figure}
\centering
  \includegraphics[scale=0.6]{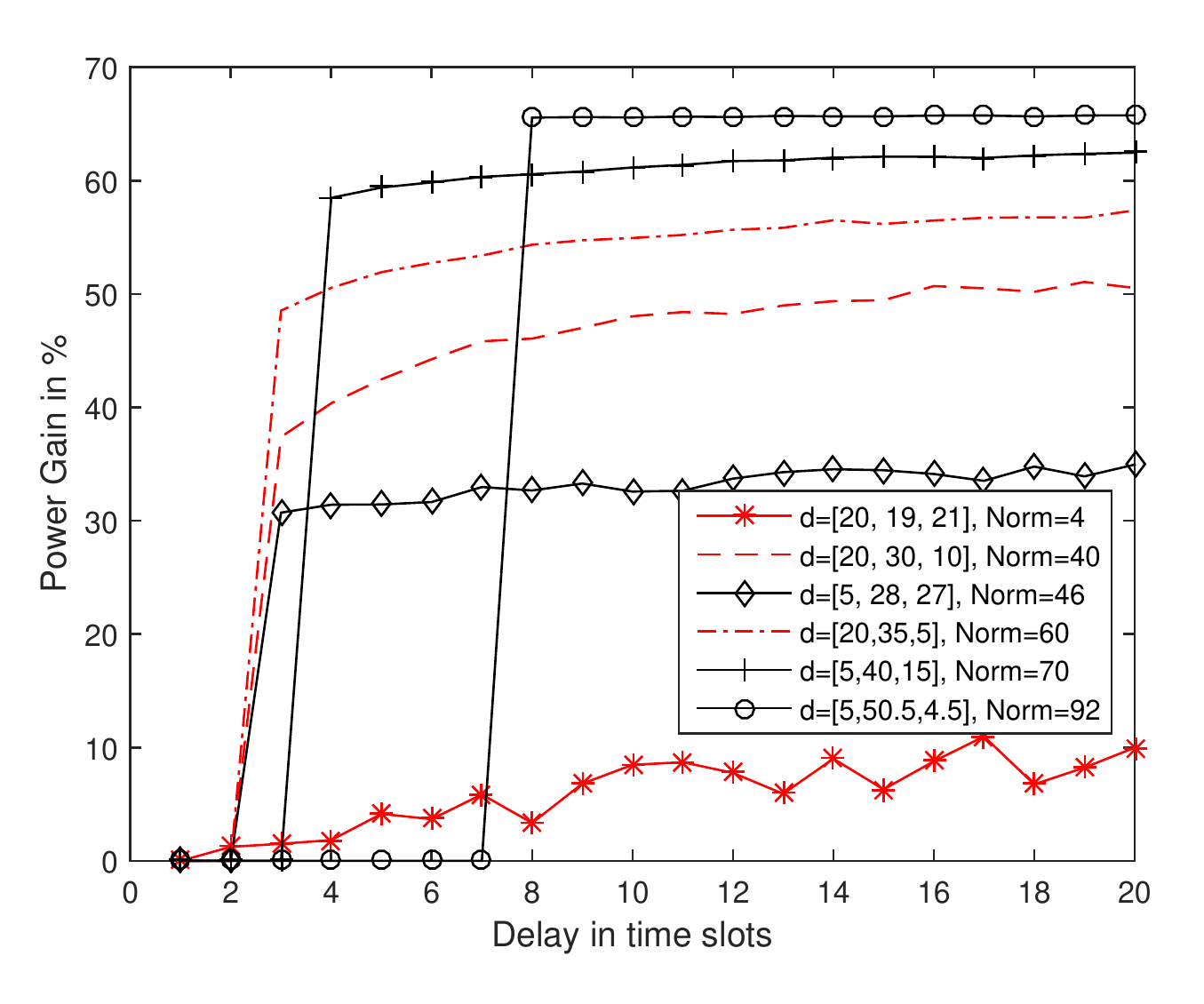}
\caption{Saving gain of the proposed power minimization algorithm in comparison to the QoS-aware scheme for an increasing target delay for various 3-hop path compositions. The target delay violation probability is fixed at $\varepsilon=10^{-3}$.}
\label{fig:shannon_qos}
\end{figure}
Now the total power consumption varies as well for the comparison case, i.e., it drops in general for the larger target delays, while it also drops for paths with more homogenous link compositions, as otherwise the critical link in strongly heterogeneous paths dominates the power consumption and delay behavior.
Therefore, in comparison to a QoS-agnostic comparison scheme, our algorithm now provides better saving gains in case of strongly heterogeneous path compositions, as only they can be significantly exploited by the proposed algorithm. 
In absolute terms, this leads to saving gains in the range of $10\%$ (in case of strongly homogeneous links) up to $70\%$ in case of strongly heterogeneous links.
Again, the power gain increases as the target delay grows.

\begin{figure}
\centering
  \includegraphics[scale=0.6]{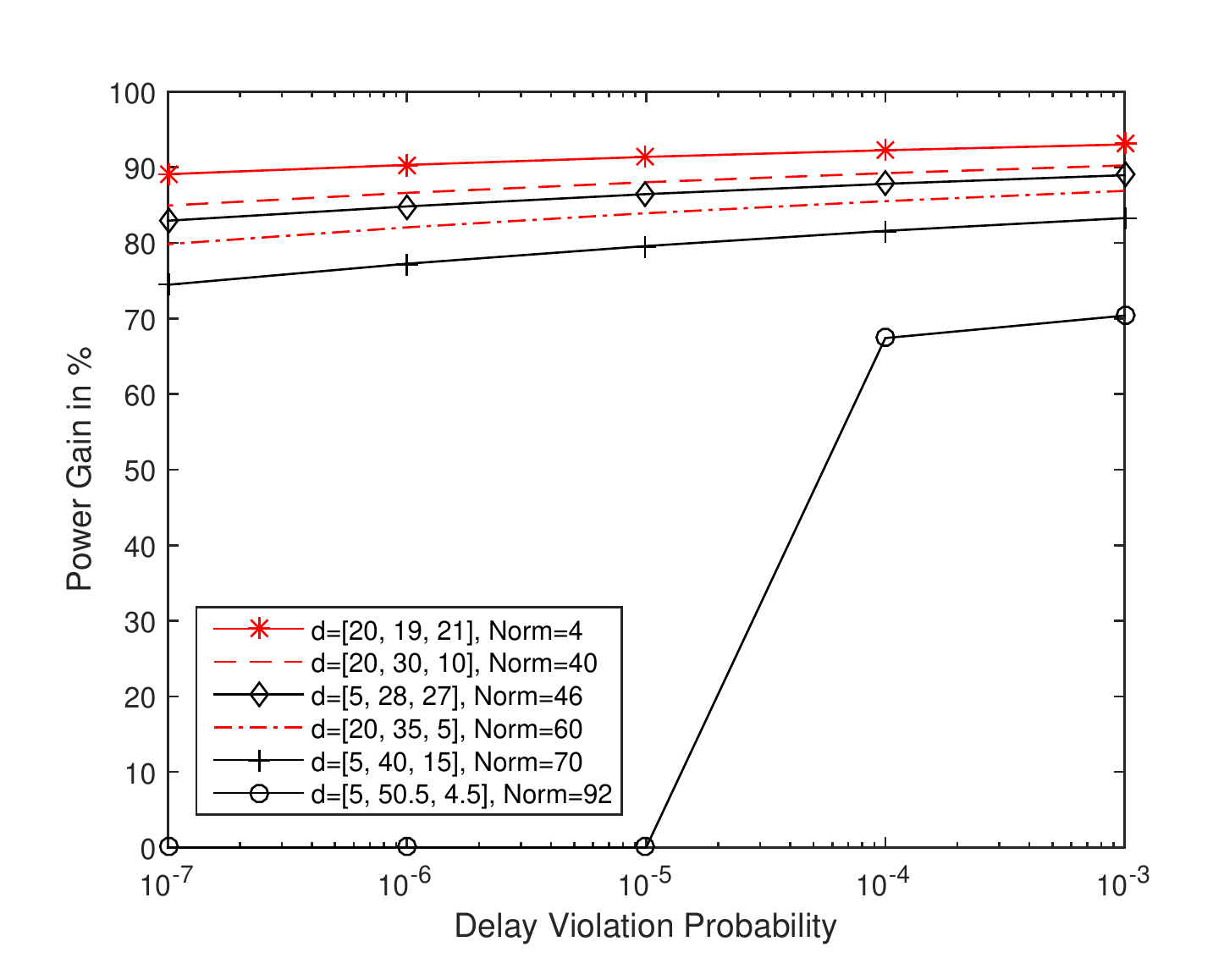}
\caption{Saving gain of the proposed power minimization algorithm in comparison to the QoS-agnostic scheme for an increasing target delay violation probability for various 3-hop path compositions.  The target delay is fixed at $w=10$ time slots.}
\label{fig:shannon_energy_gain_epsilon}
\end{figure}
Finally, in Fig.~\ref{fig:shannon_energy_gain_epsilon} we present the saving gain in comparison to the QoS-agnostic scheme in case of an increasing target delay violation probability for a fixed target delay of $w  = 10$ time slots.
We notice the same trend as in Fig.~\ref{fig:shannon_4dbm}: The smallest saving gain is around 75\% for various $\varepsilon$. 
The bigger the target violation probability, the bigger is the saving gain.
Also, the more heterogeneous the paths are, the smaller is the saving gain. 
The path with the highest norm meets the target delay for $\varepsilon \geq 10^{-4}$ with a gain of approx. 70\%.

\subsubsection{ Bound-Based vs. Simulation-Based Power-Minimization}
\label{subsec:vs_shannon} 

Although the proposed power minimization is first of its kind, it only provides a suboptimal solution since the solution is obtained by optimizing a delay bound instead of the exact delay expression which is unfortunately unattainable. Therefore, the performance gap between the obtained bound-based power allocation scheme and the real optimum, which can only be obtained using simulation, becomes relevant for our investigation. In this section, we conduct an extensive  simulation study to investigate this gap. We simulate the application in question and determine a `simulation-based' optimal power allocation. We then plot the gap between the two resulting power allocations.

For this purpose, we run  Algorithm~\ref{fun:power_min_short} with parameters $r_a=30$ bits,  $C=20$ symbols per time slot and for a target delay violation probability $\varepsilon=10^{-3}$. The resulting power allocation is set as an initial value at each simulation. Similarly as above, we observe paths of different heterogeneity with norm $R \in \{4, 46, 60, 70\}$. In each iteration the delay violation probability was obtained by a simulation. The simulation follows the same approach as defined in the algorithm: by computing the minimal gradient of the simulated delay violation probability (as in line~\ref{line:reduce_Ptx} in Algorithm~\ref{fun:power_min_short}), the link whose transmit power has to be reduced in each iteration by a certain $\deltaPtx$ is determined. Whenever no further reduction of the transmit power on any of the links is possible, i.e., the resulting delay violation probability is bigger than the target one, $\deltaPtx$ is halved. This is done maximum 15 times, i.e., $\deltaPtx$ is reduced by a factor of $2^{15}$. The initial value of $\deltaPtx=0.01~\mathrm{mW}$.
\begin{figure}
\centering
  \includegraphics[scale=0.6]{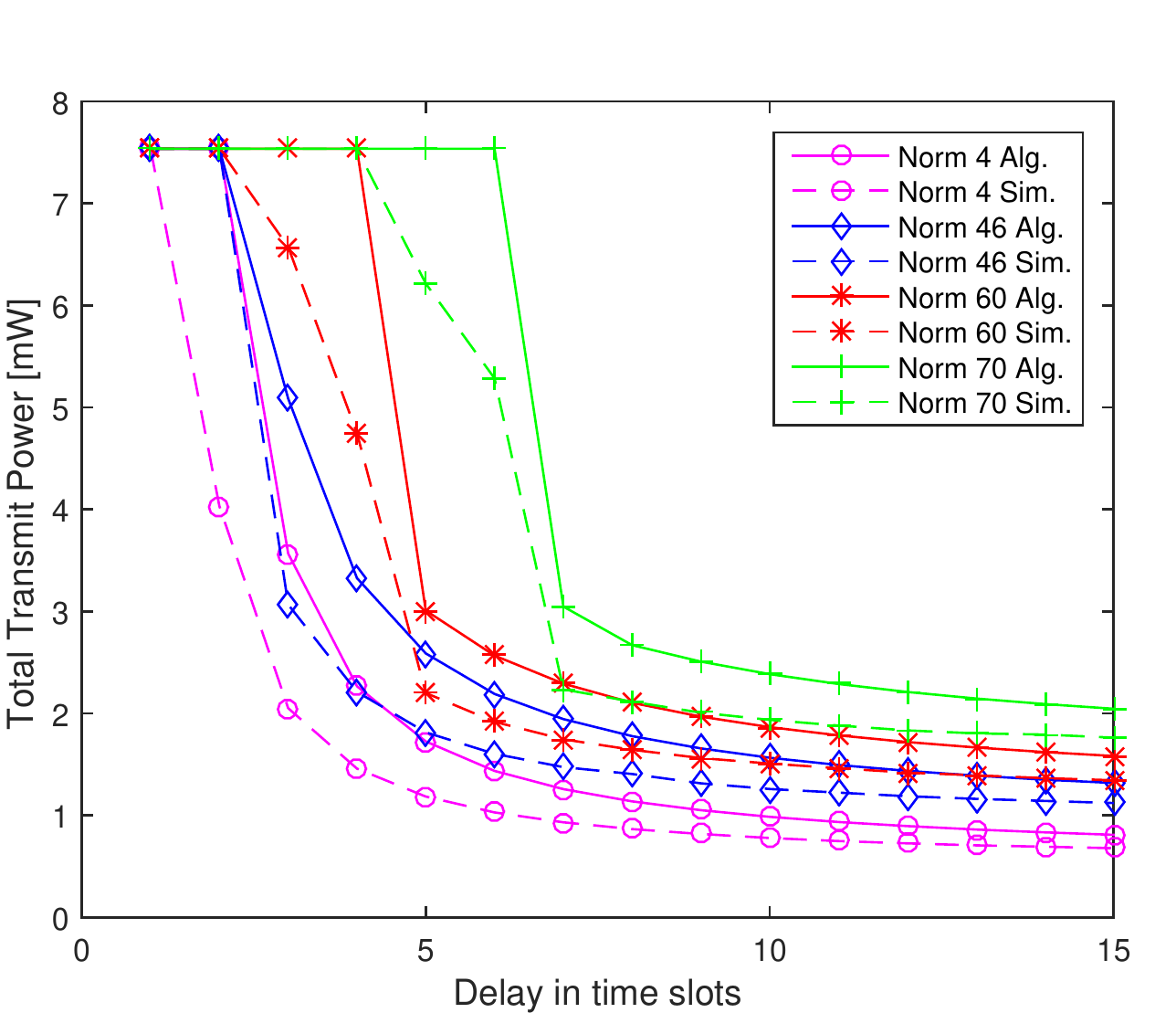}
\caption{Saving gain of the proposed power minimization algorithm in comparison to the QoS-agnostic scheme for an increasing target delay violation probability for various 3-hop path compositions.  The target delay is fixed at $w=10$ time slots.}
\label{fig:additional_total_power}
\end{figure}

In this way, we continue reducing the transmit power on the links along the path, beyond the one suggested by the algorithm. This approach enables us to characterize the gap between the total transmit power computed by the proposed  bound-based power minimization in comparison to the total transmit power obtained via simulations of the analogous process, s.t. the resulting simulated delay violation probability is close to $\varepsilon$ from below. We illustrate this additional power saving in Fig.~\ref{fig:additional_total_power} and Fig.~\ref{fig:additional_power_gain}. The former figure depicts the difference between the total transmit power obtained by the algorithm and the one obtained with simulations in absolute terms. A total transmit power of $7.5357~\mathrm{mW}$ represents the cases for which applying even the maximal possible transmit power of approx. $2.5~\mathrm{mW}$ per node does not result into meeting the desired target delay. We notice that, in absolute terms, the additional decrease of the total transmit power while performing system-based optimization is rather small for all scenarios. 
\begin{figure}
\centering
  \includegraphics[scale=0.6]{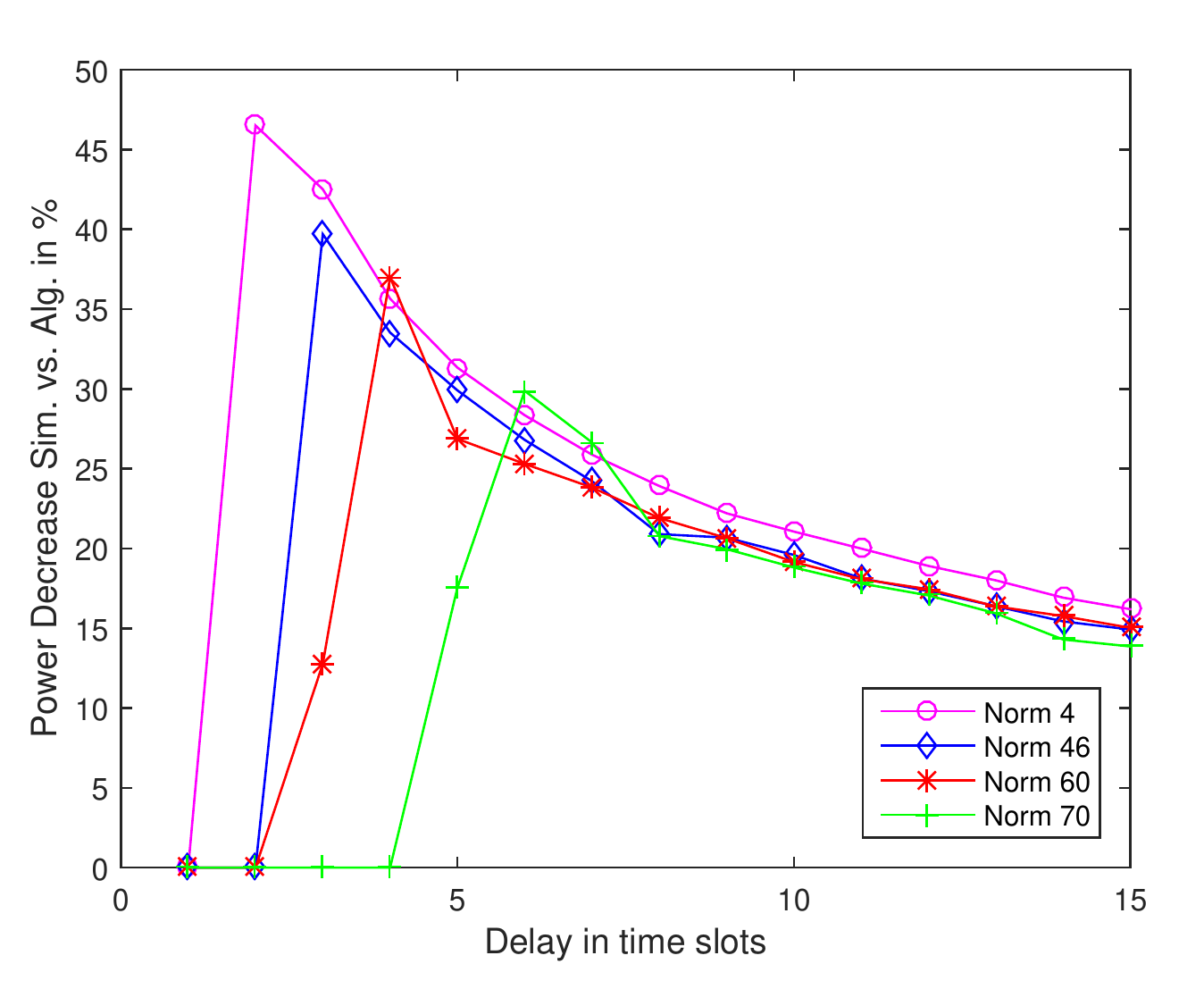}
\caption{Additional decrease of the total transmit power suggested by the  bound-based algorithm when performing simulations. The target delay violation probability is $\varepsilon=10^{-3}$.}
\label{fig:additional_power_gain}
\end{figure}
Fig.~\ref{fig:additional_power_gain} represents the additional power gain w.r.t. the total transmit power obtained from the described approach. As it can be seen, while the power gain is as big as 47\% for paths of low heterogeneity and small target delays, it converges to around 15\% for longer delays for diverse path heterogeneity. This means, that the actual system optimum, if obtained by simulations, will enable an additional 15\% decrease in the total transmit power of the path in comparison to the  bound-based power minimization. This power gap is generally significantly smaller for more heterogeneous paths. 

Note however, that, even if a further decrease of the total power by more than 15\% is observed for various scenarios, the proposed  bound-based power-minimization algorithm can be used as a good estimate on the system performance, providing at a same time delay guarantees. Moreover, applying the suggested transmit power allocation leads to an improved performance, due to the lower resulting delay violation probability. 
Such bound-based optimization is especially important for applications with strict QoS demands, introducing a so called \emph{safety gap} in the actual system performance.

\subsection{Evaluation of the Lifetime Maximization Algorithm}
\label{subsec:lifetime-max}
We evaluate in this subsection the second proposed algorithm from Sec.~\ref{subsec:lifetime_max_alg}, which takes 
the battery state of the nodes into account and maximizes the lifetime of the network by modifying the transmit power settings per node.
As this scenario and objective is more relevant in practice, we also consider a more practical channel capacity function in this section and resort to the WirelessHART industrial standard \cite{whart_online}, widely used for process automation applications with battery-powered devices. 
In order to apply our proposed algorithm, we use the provided corresponding kernel given in \cite{pe-wasun16}, defined according to the physical layer description and BER stated in the IEEE 802.15.4-2006 standard \cite{802_15_4_standard}\footnote{Note that this is no longer the active standard, since the 802.15.4-2015 is the newest version. However, the WirelessHART radios comply with IEEE 802.15.4-2006.}.
In the following subsections we first explain our methodology and then discuss some numerical results presenting insights on how QoS-aware power management can improve network lifetime under both link and battery state heterogeneity.

\subsubsection{Methodology}
\label{subsec:WirelessHART_methodology}

Let $N_s$ be the number of time slots within a superframe, while a time slot lasts for $T=10$ ms according to WirelessHART.
We hence present the delay in number of superframes, where a superframe lasts for $T\cdot N_s$ ms. 
We assume a round-robin link scheduling fashion, where the $n$-th time slot within one superframe is assigned to the $n$-th link along the path, while the channel gain varies randomly in each time slot, i.e., we assume block fading.

As shown in~\cite{pe-wasun16}, the kernel of a single-hop WirelessHART system is given by:
\begin{equation}
\label{eq:kernel_whart}
  \K(s,-w)=\frac{\left(1+ (e^{-k_as}-1)Q(\gammabar)\right)^{w} }{1-e^{r_as}\left(1+ (e^{-k_as}-1)Q(\gammabar)\right)},
\end{equation}
under the stability condition 
\begin{equation}
\label{eq:stability_whart}
\begin{aligned}
  & e^{r_as}(1+ (e^{-k_as}-1)Q(\gammabar)) < 1 \\
  \Leftrightarrow & r_a < -\frac{1}{s}\log{(1+ (e^{-k_as}-1)Q(\gammabar))}, \\
\end{aligned}
\end{equation}
\noindent where $k_a$ represents the maximal number of bits that can be transmitted in a WirelessHART time slot, $r_a$ is the size of the payload generated at the beginning of each superframe by the application and $Q(\gammabar)$ is the probability of successful MAC frame transmission over the wireless link, given as a function of the BER  and the average SNR \cite{802_15_4_standard}. 
The result in Eq.~\eqref{eq:kernel_whart} is explicitly derived in \cite{pe-wasun16} and serves as basis for the following numerical evaluation. The end-to-end kernel is obtained when the single-hop kernel is substituted into Theorem~\ref{thm:1}.

\subsubsection{Numerical Results}
\label{sec:numerics_lifetime_extension}
In the following, we are mainly interested in the network lifetime extension, represented in \%, obtained when applying our bound-based lifetime maximization algorithm (Algorithm~\ref{fun:battery_alg} in Appendix~\ref{app:pseudo_codes}) in comparison to the QoS-agnostic scheme. 
The lifetime is obtained as the minimal time duration a node can be operated by its corresponding battery among all node lifetimes for a given multi-hop path. 
The investigation is done for different target delays (in terms of superframes). 
We set $k_a = 127$ byte and the payload size $r_a$ is 10 byte \footnote{Small packets are typical for process automation applications.}. 
The arrival rate at the source node is one payload per superframe. 

To parametrize the transceiver model, we turn to the low-power Atmel IEEE802.15.4-based transceiver AT86RF233 \cite{atmel}. 
This transceiver is either in idle, send or receive mode and we obtain the power consumption in these modes from the given data sheet \footnote{Note that transceivers can offer only discrete transmit power values. Moreover, the provided data sheet contains current consumption data only for three power thresholds. For this reason, we perform polynomial curve fitting in order to obtain higher resolution consumption data and get more abstract results, not necessarily matching the transceiver capabilities in total. However, we do stay in the offered transmit power span.}.
Other system parameters are summarized in Table~\ref{tab:system_par}.  
\begin{table}[h]
\centering
\caption{System Parameters}
\label{tab:system_par}
\begin{tabular}{ c  c }
  Name & Value \\ \hline
  Total distance & 60 m \\ 
  Payload size & 10 bytes \\
  Frame size & 127 bytes \\
	Delay violation probability $\varepsilon$ & $10^{-3}$ \\
	Maximal transmit power $\Ptxmax$ & 4 dBm \\
	Current consumption in idle mode & 0.2 $\mu\mathrm{A}$ \\
	Current consumption in Rx mode & 11.8 $\mathrm{mA}$ \\
	Time slot duration & 10 ms \\
	Time spent in Tx mode & 4.256 ms \\
	ACK duration & 0.8 ms 
\end{tabular}
\end{table}

For the evaluations, we again consider different multi-hop path compositions as in Table~\ref{tab:paths}. 
However, as the battery state is another important parameter regarding the performance of the lifetime extension algorithm, we consider in addition three settings of the battery states.
In the \emph{equal} case, the battery of each node is fully charged at the moment of algorithm execution. 
The \emph{proportional} case assumes a proportional battery full state distribution among the nodes regarding the path loss on the link, i.e., the link with the highest path-loss is assigned the most charged battery. 
We finally consider the \emph{inverse proportional} battery state allocation, where the link with the highest path loss is allocated the least charged battery. All presented results refer to a target delay violation probability of $10^{-3}$ and are compared to the QoS-agnostic power allocation scheme.

\begin{figure}
\centering
  \includegraphics[scale=0.6]{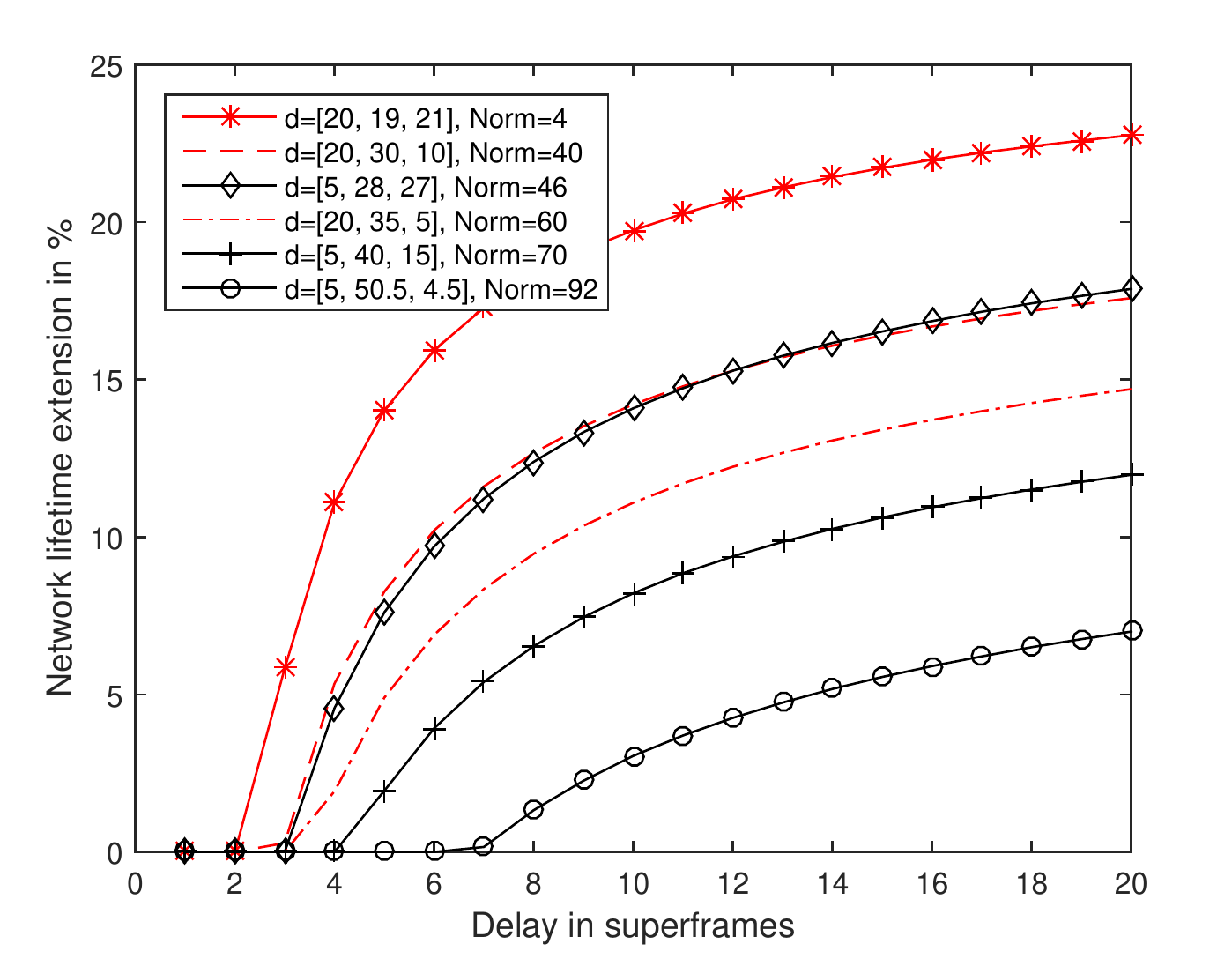}
\caption{Network lifetime extension for different target delays and path compositions. An equal battery state charge is assumed among all links.}
\label{fig:lifetime_equal_batteries}
\end{figure}
Fig.~\ref{fig:lifetime_equal_batteries} shows the gain of the network lifetime extension algorithm for different delay target $w$ considering the above discussed 3-hop path scenarios with various path norm $R$. 
In this plot, we consider the equal battery state charge among all nodes.  
For lower delays there is small gain in the network lifetime when using the algorithm in comparison to the QoS-agnostic scheme. 
As the target delay is increased, the gain in network lifetime increases. 
As we notice in Fig.~\ref{fig:lifetime_equal_batteries}, the more heterogeneous the links are, the less one can benefit from the proposed bound-based algorithm. 
The path with the least lifetime gain is the path with the biggest path norm. 

\begin{figure}
\centering
  \includegraphics[scale=0.6]{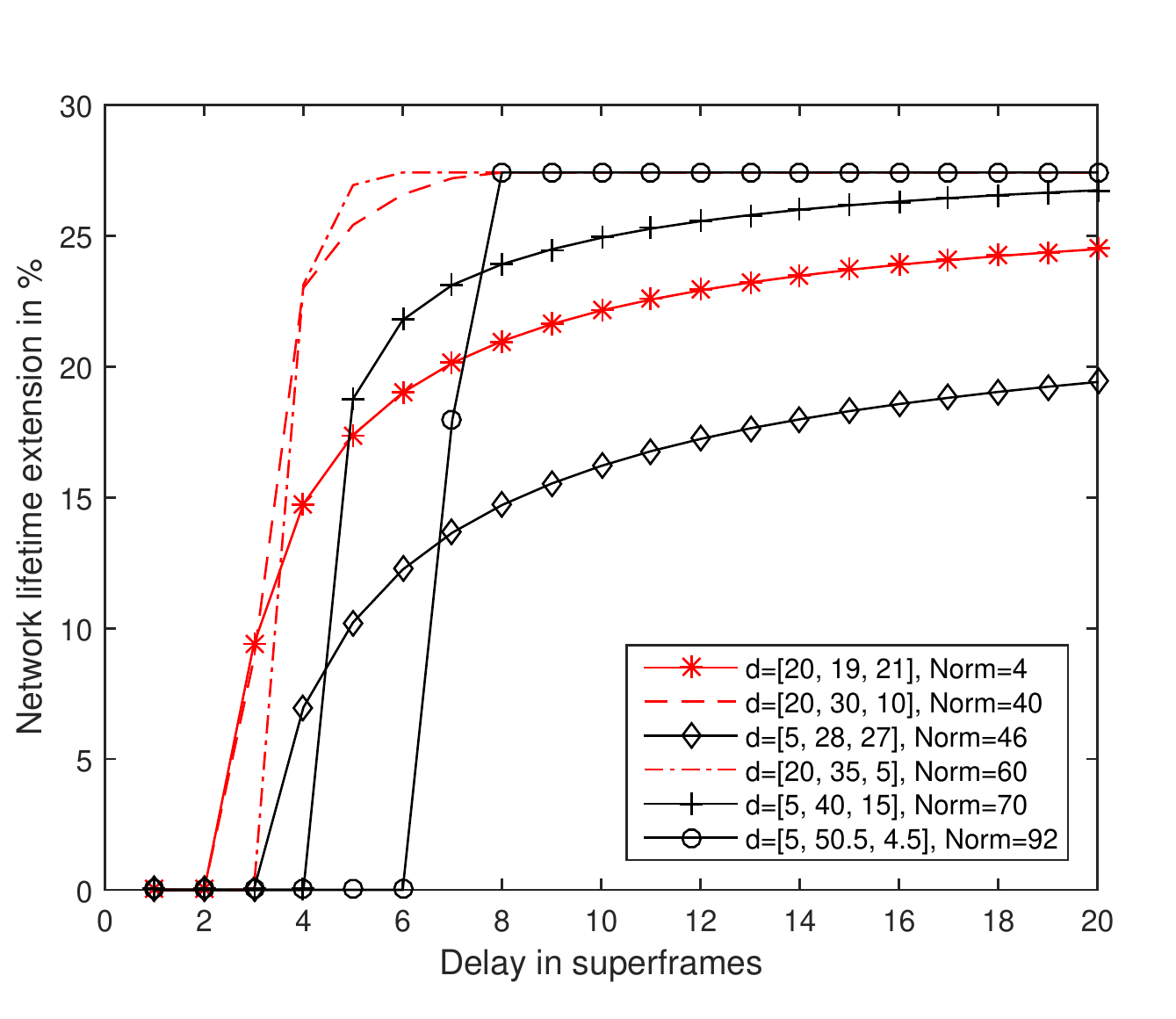}
\caption{Network lifetime extension for different target delays and path compositions. A proportional battery state charge is assumed among all links.}
\label{fig:lifetime_prop_batteries_v2}
\end{figure}
Fig.~\ref{fig:lifetime_prop_batteries_v2} illustrates the gain in network lifetime in case of proportional initial battery state distribution. 
We now notice a different trend: The paths with higher link heterogeneity benefit more from the lifetime maximization algorithm than paths with lower norm. This is due to the fact that the links with lower path loss (the better links) have a lower battery state and therefore are given advantage in the power-minimization decision, resulting finally with lower assigned transmit power and longer network lifetime. 
Note however, that the lowest lifetime extension is obtained for the path $[5, 28, 27]$, with the third link being the critical one (consuming the most energy, since it both sends and receives packets) and not the first one. 

\begin{figure}
\centering
  \includegraphics[scale=0.6]{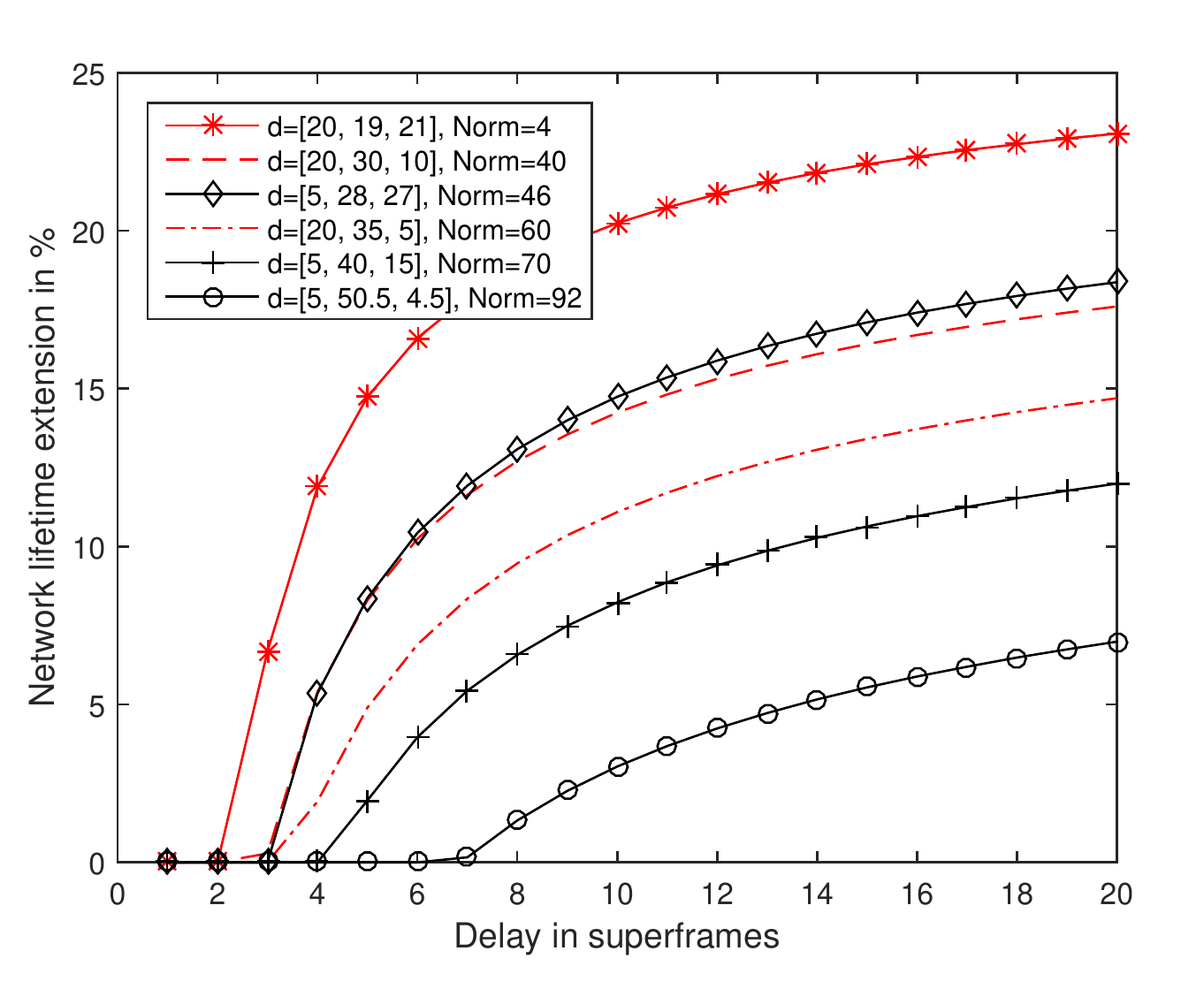}
\caption{Network lifetime extension for different target delays and path compositions. An inverse proportional battery state charge is assumed among all links.}
\label{fig:lifetime_unprop_batteries}
\end{figure}
Fig.~\ref{fig:lifetime_unprop_batteries} shows the network lifetime extension considering the same path scenarios for an inverse proportional initial battery state allocation, where the node in front of the weakest link is allocated the least battery capacity. 
We now notice the same trend as in the case of equal battery state allocation, namely that more homogeneous paths result with a larger lifetime extension. This is expected, since in both cases the algorithm prefers the links which consume more energy (i.e., the higher path loss links) when making the decision which link's transmit power to decrease.
\begin{figure}
\centering
  \includegraphics[scale=0.6]{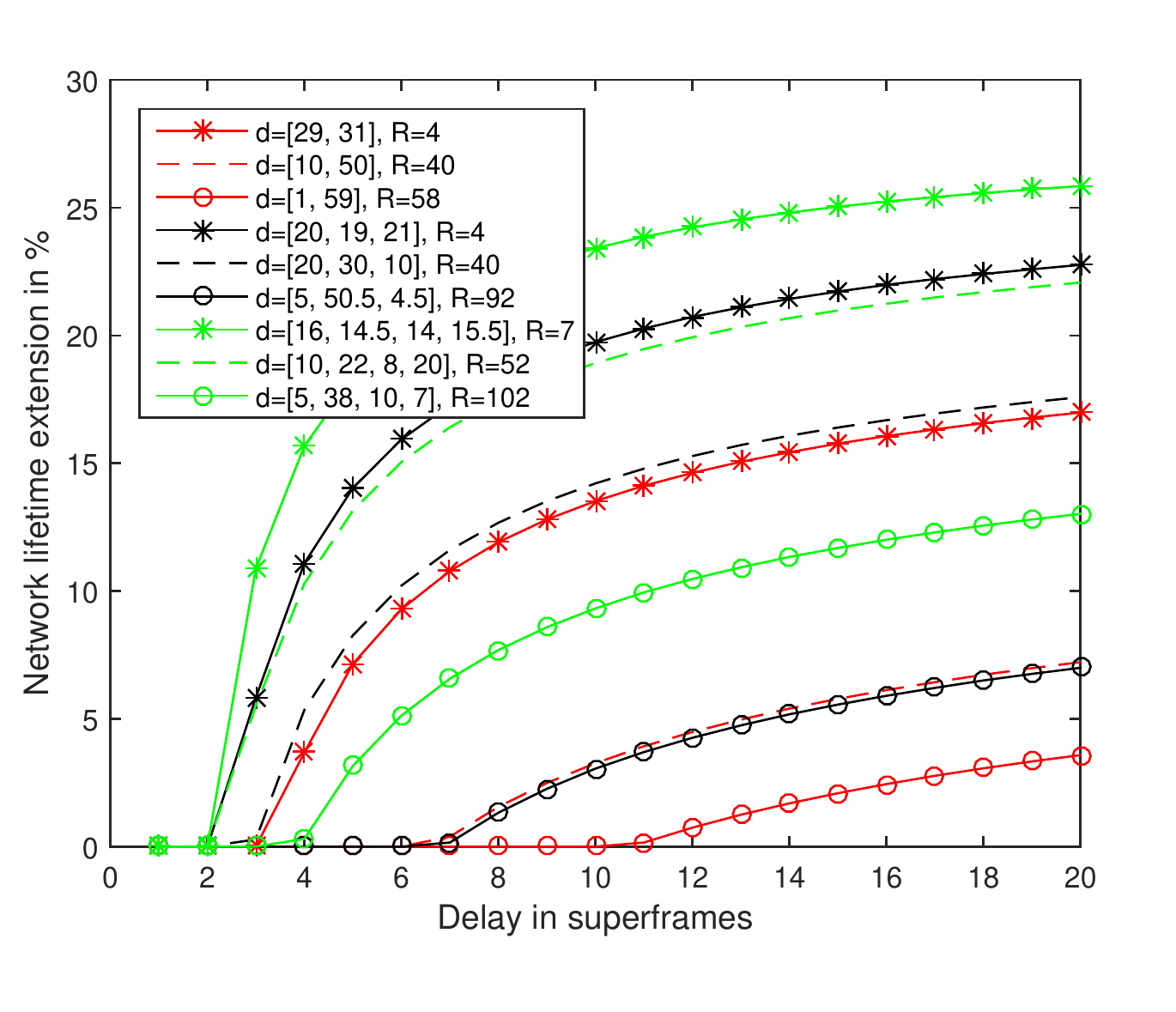}
\caption{Network lifetime extension for different target delays and path compositions for 2-, 3-, and 4-hop paths. An equal battery state charge is assumed among all links.}
\label{fig:netw_ext_more_hops}
\end{figure}

Finally, in Fig.~\ref{fig:netw_ext_more_hops} we show the network lifetime extension for 2-, 3- and 4-hop paths, each of them having low, moderate and high path norms. 
The initial battery allocation is equal among all nodes. 
We notice that the lifetime extension increases with the number of hops, however, still yielding the best one for paths with almost equal link path loss, similar to the observations in Fig.~\ref{fig:lifetime_equal_batteries} and Fig.~\ref{fig:lifetime_unprop_batteries}. 
Hence, as the path length grows, an optimal transmit power allocation under delay constraints becomes more necessary, even for rather low link heterogeneity. 

\subsubsection{Bound-Based vs. Simulation-Based Lifetime Maximization}

Similarly as in Section~\ref{subsec:vs_shannon}, we now provide a quantitative illustration of the gap between the bound-based and the simulation-based lifetime maximization. The simulation is started with a power allocation resulting from Algorithm~\ref{fun:battery_alg}, using the same system parameters as given in Table~\ref{tab:system_par}. In each iteration, the node whose transmit power is going to be reduced is chosen as the node with the least remaining battery lifetime. We are interested by how much the network lifetime can be additionally increased if using a simulation-based approach in comparison to the network lifetime provided by the bound-based method. We first take a look at the absolute values of the minimal battery duration of both schemes for lifetime maximization, shown in Fig.~\ref{fig:min_bat_duration}. The same link heterogeneity of 3-hop paths, as in Section~\ref{subsec:vs_shannon}, is considered. We notice that the minimal battery duration provided by the bound-based algorithm lies very closely to the one provided by simulations. Note that for 3-hop paths, the superframe duration is 30 ms.
\begin{figure}
\centering
  \includegraphics[scale=0.6]{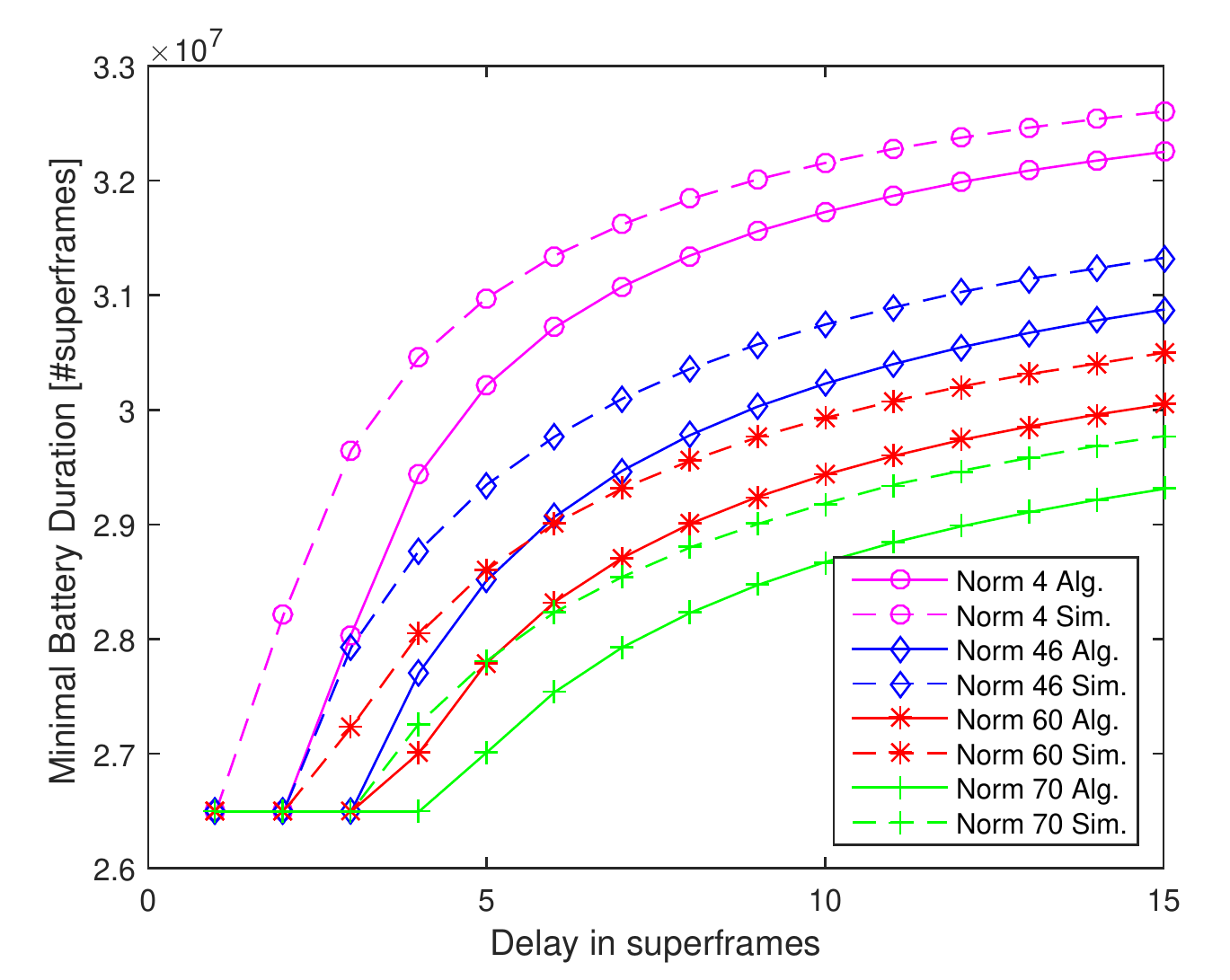}
\caption{A minimal battery duration in number of superframes vs. target delay $w$ resulting from both the bound-based and the simulation-based lifetime maximization. }
\label{fig:min_bat_duration}
\end{figure}

\begin{figure}
\centering
  \includegraphics[scale=0.6]{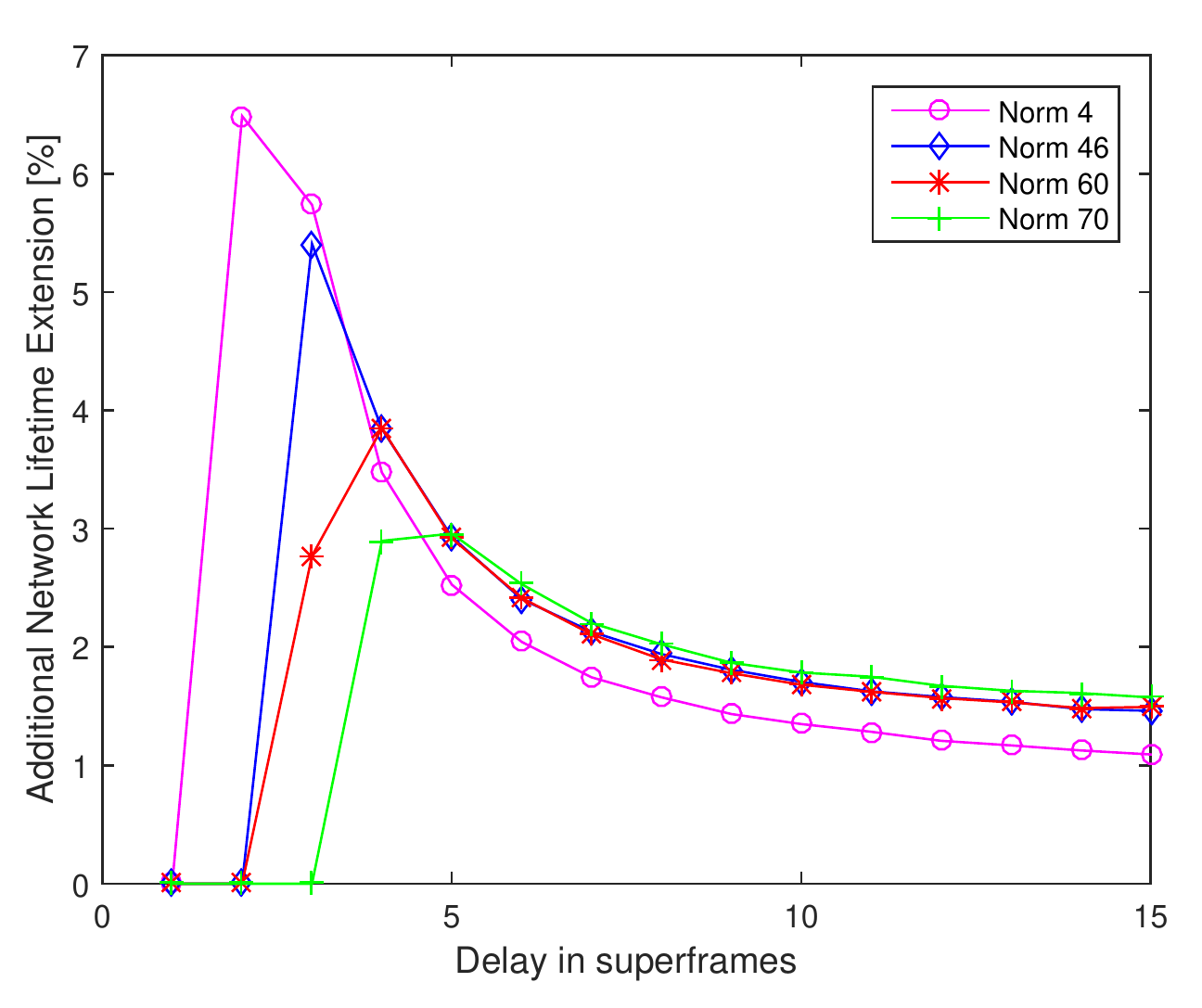}
\caption{Additional gain in the network lifetime extension for different target delays when using simulations in comparison to the delay bound-based approach.}
\label{fig:lifetime_gain}
\end{figure}
Fig.~\ref{fig:lifetime_gain} shows the additional gain in network lifetime if simulation-based instead of bound-based power optimization is applied. In comparison to the total transmit power minimization, in the WirelessHART case study we observe much smaller additional gain in the network lifetime (not bigger than 7 \%) when conducting simulations. This verifies the usefulness of the proposed algorithm for lifetime maximization, which at the same time provides performance guarantees of wireless networks employed in practice. Since the power consumed in transmit mode is only a part of the whole battery energy consumption, the optimal power allocation for a WirelessHART multi-hop network can be very closely estimated by a solution based on the end-to-end delay bound. This holds especially for paths with higher degree of link heterogeneity, as shown in the figure.

\section{Conclusion}
\label{sec:conclusion}

This paper presents a novel  bound-based algorithm for transmit power allocation with two variants: (i) total power minimization, and, (ii)  network lifetime maximization, in wireless industrial networks. 
These algorithms are based on theoretical results that we developed for the performance of heterogeneous (where both channel gain and battery full state heterogeneity are considered) multi-hop wireless networks using a stochastic network calculus approach. We provide numerical evaluation of the performance of these algorithms and compare them to simulation-based power optimization. 

In the work, we consider two wireless channel capacity models: (1) an `ideal' Shannon-capacity-based model, and, (2) a `realistic' WirelessHART-based model, the latter of which is widely applied in real industrial systems used for process automation. The numerical analysis of our proposed bound-based power optimization shows a power gain of up to 95\% for almost homogeneous multi-hop paths and more than 70\% even for paths with highly heterogeneous links using the ideal capacity model (1) above. 
It also shows that the power gain is high even for stricter delay violation probability requirements. 
A second group of numerical results, pertaining to the capacity model in (2) above, shows network lifetime extension of up to 25\% for various link heterogeneity and types of initial battery allocation strategies. We conclude that for nodes having equal battery capacities, the paths with higher link homogeneity benefit more of the proposed battery-lifetime maximization approach. However, in case of a different battery allocation strategy, 
higher lifetime extension is obtained for multi-hop paths with higher degree of link heterogeneity. 

Finally,  to evaluate the accuracy of the proposed  bound-based optimization, we provide a quantitative evaluation of the gap between the optimal power allocation resulting from the  bound-based approach  and the one obtained via simulation of the application in question.
Although the gap between the bound-based and the simulation-based optimum is non-negligible for smaller target delays in case of total transmit power minimization, there is an insignificant difference between the two approaches in case of network lifetime maximization. 
Hence, we strongly believe that the presented algorithms for bound-based optimal transmit power allocation based on stochastic network calculus principles, being first of their kind, offer a useful analytical framework for the design of wireless multi-hop networks under statistical delay constraints. 
Furthermore, the recursive nature of the devised end-to-end delay bound together with the power minimization algorithms is a solid basis for the development of  energy-efficient,  delay-aware routing algorithms for future wireless multi-hop heterogeneous networks.
 
An important extension of this work relates perhaps to the analysis of networks being exposed to multiple flows, as in meshed networks. 
Our presented approach is applicable even for such networks as long as flows utilize orthogonal resources which can be reserved per node, as is realized in WirelessHART or through approaches like Time-Sensitive Networking.
Otherwise, if the routing path is not determined, and also no resources are reserved for the flow under consideration, we provide a worst-case bound on a single flow being exposed to cross-traffic. 
This result nevertheless has to be extended towards multi-path routing, which constitutes a next step in the evolution of the framework presented here.

%
%
\section*{APPENDIX}
  \begin{appendices}
  \section{Derivation of the End-to-End Delay Bound}
	\label{app:theorem_proof}
  \begin{proof}
We start by considering the bound on the Mellin transform of the service curve of path $\mathbb{L}$ as given by Eq.~\eqref{eq:mellin_transform_path} with $i_0=\tau$ and $i_N=t$. Without loss of generality, let $m=N-1$. 
So, we obtain:
\begin{align*}
   & \M_{\S^{\mathbb{L}}}(s, t-\tau) \leq \sum_{i_1,\dots ,i_{N-1}}^t \prod_{j=1}^{N} \M_{g(\gamma_j)}^{i_j - i_{j-1}} \\ 
  & =\!\!\!\!\! \sum_{i_1, \dots ,i_{N-2}}^t \!\!\!\!\! \M_{g(\gamma_1)}^{i_1-\tau} 
      \M_{g(\gamma_2)}^{i_2-i_1} \dots \frac{ \M_{g(\gamma_N)}^{t}}{ \M_{g(\gamma_{N\!-\!1})}^{i_{N\!-\!2}}} \! \sum_{i_{N\!-\!1}=i_{N\!-\!2}}^{t} \!\!\! \left(\frac{ \M_{g(\gamma_{N\!-\!1})}}{ \M_{g(\gamma_N)}}\right)^{i_{N\!-\! 1}}  \\
  & = \sum_{i_1, \dots ,i_{N-2}}^t  \M_{g(\gamma_1)}^{i_1-\tau}  \cdot \M_{g(\gamma_2)}^{i_2-i_1}\dots  \frac{  \M_{g(\gamma_N)}^{t}}{  \M_{g(\gamma_{N-1})}^{i_{N-2}}} \\
  &  \cdot \left(\frac{\left(\frac{  \M_{g(\gamma_{N-1})} }{  \M_{g(\gamma_N)}}\right)^{i_{N-2}}-\left(\frac{  \M_{g(\gamma_{N-1})}}{  \M_{g(\gamma_N)}}\right)^{t + 1}}{1-\frac{  \M_{g(\gamma_{N-1})}}{  \M_{g(\gamma_N)}}}\right)   \\
  & = \frac{  \M_{g(\gamma_N)} }{ \M_{g(\gamma_N)} - \M_{g(\gamma_{N-1})}}  \sum_{i_1, \dots ,i_{N-2}}^t  \M_{g(\gamma_1)}^{i_1-\tau}\dots \frac{ \M_{g(\gamma_N)}^{t}}{ \M_{g(\gamma_{N-1})}^{i_{N-2}}}\\
  & \cdot \left(\frac{ \M_{g(\gamma_{N-1})}}{ \M_{g(\gamma_N)}}\right)^{i_{N-2}}  - \frac{ \M_{g(\gamma_N)}}{ \M_{g(\gamma_{N})} - \M_{g(\gamma_{N-1})}}  \\
  & \sum_{i_1, \dots ,i_{N-2}}^t 
     \M_{g(\gamma_1)}^{i_1-\tau}\dots \frac{{\M_{g(\gamma_N)}}^{t}}{ \M_{g(\gamma_{N-1})}^{i_{N-2}}}\left(\frac{ \M_{g(\gamma_{N-1})}}{ \M_{g(\gamma_N)}}\right)^{t} \frac{ \M_{g(\gamma_{N-1})}}{ \M_{g(\gamma_N)}}   \\
  & = \frac{ \M_{g(\gamma_N)}}{ \M_{g(\gamma_N)} \!- \! \M_{g(\gamma_{N-1})}} \!\!  \sum_{i_1, \dots ,i_{N-2}}^t \!\!\! \!\! \M_{g(\gamma_1)}^{i_1-\tau} \dots \M_{g(\gamma_{N-2})}^{i_{N-3}-i_{N-2}} \M_{g(\gamma_N)}^{t - i_{N-2}} \\
  & -\! \frac{ \M_{g(\gamma_{N-1})}}{ \M_{g(\gamma_N)} \!- \! \M_{g(\gamma_{N-1})}} \!\!  \sum_{i_1, ... ,i_{N-2}}^t \!\!\!\!\! \M_{g(\gamma_1)}^{i_1-\tau} \dots \M_{g(\gamma_{N-2})}^{i_{N-3}-i_{N-2}} \M_{g(\gamma_{N-1})}^{t -i_{N-2}} \\ 
  & = \frac{ \M_{g(\gamma_{N})}} { \M_{g(\gamma_N)} - \M_{g(\gamma_{N-1}) }} \M_{\S^{\mathbb{L} \setminus \{N-1\}}}(t-\tau) \\
  & + \frac{ \M_{g(\gamma_{N-1}) } }{ \M_{g(\gamma_{N-1})} - \M_{g(\gamma_{N})} } \M_{\S^{\mathbb{L}\setminus \{N\}}}(t-\tau) \, ,
 \end{align*}
where we have omitted for readability that all Mellin transforms are functions of $s$ above.
Thus, we have shown that an upper bound of the Mellin transform of path $\mathbb{L}$ can be obtained recursively from the Mellin transform of the service process of paths $\mathbb{L} \setminus \{N-1\}$ and  $\mathbb{L} \setminus \{N\}$.  
 
As the kernel is a function of the Mellin transforms of the SNR domain arrival and service process, i.e.,
\begin{equation} 
\begin{aligned}
\Kd^{\mathbb{L}} \left(s, t+w, t\right) \!=\! \sum_{i=0}^t \M_{\A}\left(1+s,i,t\right)    \mathcal{M}_{\mathcal{S}^{\mathbb{L}}}(1-s,i,t+w) \, , \nonumber
 \label{eq:proof_1}
 \end{aligned}
\end{equation}
 it follows directly that the steady state kernel $\Kd^{\mathbb{L}} \left(s,-w\right)$ is a recursive function of the kernels $\Kd^{\mathbb{L} \setminus \{N-1\}} \left(s,-w \right)$ and $\Kd^{\mathbb{L} \setminus \{N\}} \left(s,-w \right)$ as claimed in the theorem.
\end{proof}

\section{Pseudo Codes of the Bound-Based Algorithms}
\label{app:pseudo_codes}

In this appendix we present the pseudo-code for the bound-based power-minimization algorithm (Algorithm \ref{fun:power_min_short}) and the bound-based network lifetime extension algorithm (Algorithm \ref{fun:battery_alg}). The function \emph{search\_s} (Algorithm \ref{fun:search_s}) is an auxiliary function called in both algorithms. 

\begin{algorithm}[ht]
\caption{Search $s^{*} \in (0,b)$ for which $\K^{\L}(s^{*},-w)$ is minimal}
\label{fun:search_s}
\begin{algorithmic}[1]
\Function{SEARCH\_S }{$0,b,\gammabar,\Deltamin,r_a,w$}
\Ensure Find $s^{*}$

\State Compute ${s_\mathrm{l},s_\mathrm{m},s_\mathrm{r}} \in (0,b)$
\State $s_{\mathrm{start}}=0, s_{\mathrm{end}}=b$; use for simplicity $\K^{\L}(s_i,-w) \triangleq \K^{\L}(s_i)$

\If {$s_\mathrm{m}-s_\mathrm{l} > \Deltamin$}
  \State Find out in which interval lies $s^{*}$ 
  \Statex ***Case 1: $s^{*} \in (s_{\mathrm{start}},s_\mathrm{m})$
  \If {$\K^{\L}(s_{\mathrm{end}})>\K^{\L}(s_\mathrm{r})>\K^{\L}(s_\mathrm{m})>\K^{\L}(s_\mathrm{l})$}
     \State $s^{*}= \text{search\_s}(s_{\mathrm{start}},s_\mathrm{m},\gammabar,\Deltamin,\kts,w)$
  \Statex ***Case 2: $s^{*} \in (s_{\mathrm{m}},s_\mathrm{end})$
  \ElsIf {$\K^{\L}(s_{\mathrm{start}})>\K^{\L}(s_\mathrm{l})>\K^{\L}(s_\mathrm{m})>\K^{\L}(s_\mathrm{r})$}
     \State $s^{*}= \text{search\_s}(s_\mathrm{m},s_{\mathrm{end}},\gammabar,\Deltamin,\kts,w)$
  \Statex ***Case 3: $s^{*} \in (s_{\mathrm{l}},s_\mathrm{r})$
  \ElsIf {$\K^{\L}(s_{\mathrm{end}})>\K^{\L}(s_\mathrm{r})>\K^{\L}(s_\mathrm{m})$
  \State AND $\K^{\L}(s_\mathrm{start})>\K^{\L}(s_\mathrm{l})>\K^{\L}(s_{\mathrm{m}})$}
     \State $s^{*}= \text{search\_s}(s_\mathrm{l},s_\mathrm{r},\gammabar,\Deltamin,\kts,w)$
    \EndIf
\Else 
  \State $s^{*}=s_\mathrm{m}$
  \State \Return
\EndIf
\EndFunction
\end{algorithmic}
\end{algorithm}


\begin{algorithm}[ht]
\caption{Transmit Power-Minimization Algorithm}
\label{fun:power_min_short}
\begin{algorithmic}[1]
\Require $|\vec{{\bar{h}}}^2|, \deltaPtx, \Ptxmax, \deltaPtxmin, r_a, w, \varepsilon, \deltaEpsilon$
\Ensure $\min \sum_{n=1}^N \Ptx_n$, s.t. $\K^{\L}(s,-w) \leq \varepsilon$

\State $\vPtx=\Ptxmax$; $\gammabar=\min \{{\vec{\gammabar}}_{\mathrm{max}}\}=\min[f(\Ptxmax, |\vec{{\bar{h}}}^2|]$; $\K^{\L}(s,-w) \triangleq \K^{\L}(s,\vPtx)$

\State Find $s'$, $\forall s \in (0,s'),$ s.t. Eq.~\eqref{eq:stability_cond} holds for $\gammabar$
\State Compute $\vec{\mathrm{P}}_{\text{txmin}}$, s.t. channel capacity $ \geq r_a$
\State Compute current delay bound $\hat{\varepsilon}=\K^{\L}(s, \vPtx)$
%
\If {$(\hat{\varepsilon} > \varepsilon)$} \Return fail
\Else
   \While {$(\hat{\varepsilon} \notin (\varepsilon-\deltaEpsilon,\varepsilon))$}
       \While {$(\hat{\varepsilon} > \varepsilon)$ AND ($\deltaPtx > \deltaPtxmin$)}
          \State Choose smaller $\deltaPtx: \deltaPtx=\nicefrac{\deltaPtx}{2}$
			\EndWhile
          \State $\vPtxprim=\vPtx$; ${\vPtx}_{n}$ is transmit power vector where $p_n={\Ptx}_n'-\deltaPtx$
          \label{line:gradient}\State Find the smallest gradient: $n=\underset{N}{\operatorname{argmin}} \triangledown {\K}_n= \lvert \frac{\hat{\varepsilon}-\K^{\L}(s, \vPtx_n)}{\deltaPtx}\rvert$
          \label{line:reduce_Ptx}\State $\Ptxn'=\Ptxn-\deltaPtx$, assure $\Ptxn' \geq {\Ptxmin}_n$; $\vgammabar=f(\vPtxprim,|\vec{{\bar{h}}}^2|)$ 
          \State $s^{*}=\text{search\_s}(s_{\mathrm{start}},b,\vgammabar,\Deltamin,r_a,w)$; compute $\hat{\varepsilon}=\K^{\L}(s^{*}, \vPtxprim)$
       
       \State $\vPtx=\vPtxprim$
   \EndWhile
	 \State \Return $\vPtx$
\EndIf

\end{algorithmic}
\end{algorithm}

\begin{algorithm}[ht]
\caption{Network Lifetime Maximization Algorithm}
\label{fun:battery_alg}
\begin{algorithmic}[1]
\Require $|\vec{{\bar{h}}}^2|, \deltaPtx, \Ptxmax, \deltaPtxmin, k_a, r_a, w, \varepsilon, \deltaEpsilon, {\vec{\mathrm{B}}}$
\Ensure $\max \underset{m}{\min} \{\theta_m\}, \theta_m \in \vec{\mathrm{\theta}}, \vec{\mathrm{\theta}}=\frac{\vec{\mathrm{B}}}{{T\vPtx}}$, $\K^{\L}(s, -w) \leq \varepsilon$

\State $\vPtx=\Ptxmax$; $\gammabar=\min \{{\vec{\gammabar}}_{\mathrm{max}}\}$
\State Find $s'$, s.t. $\forall s \in (0,s')$ stability condition Eq.~\eqref{eq:stability_cond} holds for $\gammabar$
\State Set ${\mathrm{P}}_{\text{txmin}}$ to transceiver capabilities; $\hat{\varepsilon}=\K^{\L}(s,-w)$
\State Same approach as in Algorithm~\ref{fun:power_min_short}, only replace line~\ref{line:gradient} and~\ref{line:reduce_Ptx} with
\State \label{line:ptx_min_criteria}$M=\forall m: \{\underset {m}{\operatorname{argmin}} \{\theta_m\}\}$, $\forall m \in M: \Ptx_m'={\Ptx}_m-\frac{\deltaPtx}{\lvert M \rvert}$, assure ${\Ptx}_m' \geq {\Ptxmin}$ 

\end{algorithmic}
\end{algorithm}



\FloatBarrier
\section{Proof of the Delay Bound Convexity}
\label{app:convexity}
\begin{proof}
As shown in Eq.~\eqref{eq:delay_kernel}, the kernel has the following form:
\begin{equation}
\label{eq:kernel_ab}
\begin{aligned}
\Kd\left(s,-w\right) & = \frac{\M_{g\left(\gamma\right)}^w\left(1 - s\right)}{1 - \M_{\alpha}\left(1 + s \right) \M_{g\left(\gamma\right)}\left(1 - s \right)} \\
& = \frac{(\M_{\S}(1-s,\tau,t))^w}{1 - \M_{\A}(1+s,\tau,t) \M_{\S}(1-s,\tau,t)},
\end{aligned}
\end{equation}

\noindent where $\M_{\A}(1+s,\tau,t)$ is the Mellin transform of the arrival process in $1+s$ and $\M_{\S}(1-s,\tau,t)$ is the Mellin transform of the service in $1-s$.
From its definition it follows that, \mbox{$\M_{\A}(1+s,\tau,t)$} is strictly increasing and $\M_{\S}(1-s,\tau,t)$ is strictly decreasing in $s$. Note that we limit in the following the scope of the parameter $s$ to the stability region, i.e., $s\in(0,b)$. 
In order to prove that Eq.~\eqref{eq:kernel_ab} is convex, we will first show that \mbox{$\M_{\A}(1+s,\tau,t) \cdot \M_{\S}(1-s,\tau,t)$} is convex. We will then use this fact to proof the convexity of the whole kernel. 

Let us consider the following functions:

\begin{equation}
\begin{aligned}
& \M_{\A}(1+s,\tau,t) = \left(\EE[e^{s\alpha}]\right)^{t-\tau}\\
& \M_{\S}(1-s,\tau,t) = \left(\EE[e^{-s\beta}]\right)^{t-\tau},
\end{aligned}
\end{equation}

\noindent where $\alpha$ and $\beta$ are the random arrival to the link and service offered by the link. Note here that, $\left(\EE[e^{-s\beta}]\right)^{t-\tau}\leq 1$ for $s\in(0,b)$, i.e., the Mellin transform of the service reaches its maximum for $s=0$. In the following we will prove that $\left(\EE[e^{s\alpha}]\right)^{t-\tau} \cdot \left(\EE[e^{-s\beta}]\right)^{t-\tau}$ is convex in $s$ for the range of $s$ for which the stability condition holds.

We start by rewriting the product as
\begin{equation}
\label{eq:stab_cond2}
\left(\mathbb{E}[e^{s\alpha}]\right)^{t-\tau} \cdot \left(\EE[e^{-s\beta}]\right)^{t-\tau}=\left(\EE[e^{-sx}]\right)^{t-\tau}<1,
\end{equation}

\noindent where we substitute $x=\beta-\alpha$ and use the independence of $\alpha$ and $\beta$. In order for $\left(\EE[e^{-sx}]\right)^{t-\tau}$ to be convex, it has to hold for any $0\leq \delta \leq 1$:
\begin{equation}
\label{eq:convex_condition}
\left(\EE[e^{-(\delta s_1+(1-\delta)s_2)x}]\right)^{t-\tau} 
 \leq \delta\left(\EE[e^{-s_1x}]\right)^{t-\tau} + (1-\delta)\left(\EE[e^{-s_2x}]\right)^{t-\tau}.
\end{equation}

Let us apply H\"older's inequality \cite{hoelder} to the left-hand side of Eq.~\eqref{eq:convex_condition}. H\"older's inequality states that
\begin{equation}
\EE[|X\cdot Y|]\leq \left(\EE|X|^p\right)^{\nicefrac{1}{p}}\cdot\left(\EE|Y|^q\right)^{\nicefrac{1}{q}}
\end{equation}
for $\forall p,q$ for which $\frac{1}{p}+\frac{1}{q}=1$. Hence, for $\frac{1}{p}=\delta$ and $\frac{1}{q}=1-\delta$, we have:
\begin{equation}
\label{eq:h1}
\begin{aligned}
&\left(\EE[e^{-\left(\delta s_1+(1-\delta)s_2\right)x}]\right)^{t-\tau}=\left(\EE[|e^{-\delta s_1x}\cdot e^{-(1-\delta)s_2x}|]\right)^{t-\tau}\\
& \leq \left(\EE[|e^{-\delta s_1x}|^{\nicefrac{1}{\delta}}]\right)^{(t-\tau)\delta} \cdot \left(\EE[|e^{-(1-\delta) s_2x}|^{\nicefrac{1}{1-\delta}}]\right)^{(t-\tau)(1-\delta)} \\
& = \left(\EE[e^{-s_1x}]\right)^{(t-\tau)\delta} \cdot \left(\EE[e^{-s_2x}]\right)^{(t-\tau)(1-\delta)} \\
& \leq \delta \left(\EE[e^{-s_1x}]\right)^{t-\tau} + (1-\delta) \left(\EE[e^{-s_2x}]\right)^{t-\tau}.
\end{aligned}
\end{equation}
\noindent In the last line we use the fact that
\begin{equation}
\label{eq:temp_inequality}
u^{\delta} v^{1-\delta} \leq \delta u+(1-\delta)v, \forall \delta \in [0,1],
\end{equation}
which holds because of the following: Let us define $f(\delta)=u^{\delta} v^{1-\delta} - \delta u-(1-\delta)v$. The second derivative of $f(\delta)$ results to:
\begin{equation}
f''(\delta)=u^{\delta} v^{1-\delta} (\log{(u)} - \log{(v)})^{2} \geq 0,
\end{equation}
\noindent $\forall \delta \in [0,1]$ and $u,v < 1$. Since $f(0)=f(1)=0$ and $f''(\delta) \geq 0$, it follows that $f(\delta)$ reaches a local minimum for $\delta \in [0,1]$. Hence, $f(\delta) \leq 0, \forall \delta \in [0,1]$ and therefore Eq.~\eqref{eq:temp_inequality} holds.
Therefore, we show that $\M_{\A}(1+s,\tau,t) \M_{\S}(1-s,\tau,t)$ is convex.
Having shown this, it follows that the function $1-\M_{\A}(1+s,\tau,t) \M_{\S}(1-s,\tau,t)$ is concave and positive, since \mbox{$\M_{\A}(1+s,\tau,t) \M_{\S}(1-s,\tau,t) \leq 1$} due to the stability condition. Hence, the reciprocal is convex, i.e., $\frac{1}{1-\M_{\A}(1+s,\tau,t) \M_{\S}(1-s,\tau,t)}$ is convex \cite{Book:boyd}. 

We now turn to the second part of the proof, namely, we consider the entire kernel. In order to show that the kernel given with Eq.~\eqref{eq:kernel_ab} is convex, it has to hold:

\begin{equation}
\label{eq:whole_convexity}
\begin{aligned}
 & \frac{\left(\EE[e^{-(\delta s_1+(1-\delta)s_2)\beta}]\right)^{(t-\tau)w}}{1-\left(\EE[e^{(\delta s_1+(1-\delta)s_2)\alpha}]\cdot\EE[e^{-(\delta s_1+(1-\delta)s_2)\beta}]\right)^{t-\tau}} \\
& \leq \delta \cdot \frac{\left(\EE[e^{-s_1\beta}]\right)^{(t-\tau)w}}{1-\left(\EE[e^{s_1\alpha}]\cdot\EE[e^{-s_1\beta}]\right)^{t-\tau}} + (1-\delta) \cdot  \frac{\left(\EE[e^{-s_2\beta}]\right)^{(t-\tau)w}}{1-\left(\EE[e^{s_2\alpha}]\cdot\EE[e^{-s_2\beta}]\right)^{t-\tau}},
\end{aligned}
\end{equation}
$\forall s_1,s_2 \in (0,b)$ and $0 \leq \delta \leq 1$. Since $\frac{1}{1-\M_{\A}(1+s,\tau,t) \cdot \M_{\S}(1-s,\tau,t)}$ is convex, we know:
\begin{equation}
\label{eq:ma_ms_convex}
\begin{aligned}
 & \frac{1}{1-\left(\EE[e^{(\delta s_1+(1-\delta)s_2)\alpha}]\cdot\EE[e^{-(\delta s_1+(1-\delta)s_2)\beta}]\right)^{t-\tau}} \\
& \leq \delta \cdot \frac{1}{1-\left(\EE[e^{s_1\alpha}]\cdot\EE[e^{-s_1\beta}]\right)^{t-\tau}} + (1-\delta) \cdot \frac{1}{1-\left(\EE[e^{s_2\alpha}]\cdot\EE[e^{-s_2\beta}]\right)^{t-\tau}}.
\end{aligned}
\end{equation}
Multiplying the left- and right-hand side of Eq.~\eqref{eq:ma_ms_convex} by $\left(\EE[e^{-(\delta s_1+(1-\delta)s_2)\beta}]\right)^{(t-\tau)w}$, we obtain:
\begin{equation}
\label{eq:eq2}
\begin{aligned}
 & \frac{\left(\EE[e^{-(\delta s_1+(1-\delta)s_2)\beta}]\right)^{(t-\tau)w}}{1-\left(\EE[e^{(\delta s_1+(1-\delta)s_2)\alpha}]\cdot\EE[e^{-(\delta s_1+(1-\delta)s_2)\beta}]\right)^{t-\tau}} \\
& \leq \delta \cdot \frac{\left(\EE[e^{-(\delta s_1+(1-\delta)s_2)\beta}]\right)^{(t-\tau)w}}{1-\left(\EE[e^{s_1\alpha}]\cdot\EE[e^{-s_1\beta}]\right)^{t-\tau}} + (1-\delta) \cdot \frac{\left(\EE[e^{-(\delta s_1+(1-\delta)s_2)\beta}]\right)^{(t-\tau)w}}{1-\left(\EE[e^{s_2\alpha}]\cdot\EE[e^{-s_2\beta}]\right)^{t-\tau}}.
\end{aligned}
\end{equation}
Now, using again H\"older's inequality for $\frac{1}{p}=\delta$ and $\frac{1}{q}=1-\delta$, we obtain:
\begin{equation}
\begin{aligned}
&\left(\EE[e^{-(\delta s_1+(1-\delta)s_2)\beta}]\right)^{(t-\tau)w} = \left(\EE[|e^{-\delta s_1 \beta}| \cdot |e^{-(1-\delta)s_2 \beta}|]\right)^{(t-\tau)w} \\
& \leq \left(\EE[|e^{-\delta s_1 \beta}|^{\nicefrac{1}{\delta}}]\right)^{(t-\tau)w\delta} \cdot \left(\EE[|e^{-(1-\delta) s_2 \beta}|^{\nicefrac{1}{1-\delta}}]\right)^{(t-\tau)w(1-\delta)} \\
& = \left(\EE[|e^{-s_1 \beta}|]\right)^{(t-\tau)w\delta} \cdot \left(\EE[|e^{-s_2 \beta}|]\right)^{(t-\tau)w(1-\delta)} \\
& \leq \left(\EE[|e^{-s_1 \beta}|]\right)^{(t-\tau)w} \cdot \left(\EE[|e^{-s_2 \beta}|]\right)^{(t-\tau)w},
\end{aligned}
\end{equation}
since $\left(\EE[|e^{-s \beta}|]\right)^{(t-\tau)w} > 0$.
Hence, the following inequality holds for the right-hand side of Eq.~\eqref{eq:eq2}:
\begin{equation}
\label{eq:eq3}
\begin{aligned}
& \delta \cdot \frac{\left(\EE[e^{-(\delta s_1+(1-\delta)s_2)\beta}]\right)^{(t-\tau)w}}{1-\left(\EE[e^{s_1\alpha}]\cdot\EE[e^{-s_1\beta}]\right)^{t-\tau}} + (1-\delta) \cdot \frac{\left(\EE[e^{-(\delta s_1+(1-\delta)s_2)\beta}]\right)^{(t-\tau)w}}{1-\left(\EE[e^{s_2\alpha}]\cdot\EE[e^{-s_2\beta}]\right)^{t-\tau}} \\
& \leq \delta \cdot \frac{\left(\EE[|e^{-s_1\beta}|] \cdot \EE[|e^{-s_2\beta}|]\right)^{(t-\tau)w}}{1-\left(\EE[e^{s_1\alpha}]\cdot\EE[e^{-s_1\beta}]\right)^{t-\tau}} + (1-\delta) \cdot \frac{\left(\EE[|e^{-s_1\beta}|] \cdot \EE[|e^{-s_2\beta}|]\right)^{(t-\tau)w}}{1-\left(\EE[e^{s_2\alpha}]\cdot\EE[e^{-s_2\beta}]\right)^{t-\tau}}\\
& \leq \delta \cdot \frac{\left(\EE[|e^{-s_1\beta}|]\right)^{(t-\tau)w}}{1-\left(\EE[e^{s_1\alpha}]\cdot\EE[e^{-s_1\beta}]\right)^{t-\tau}} + (1-\delta) \cdot \frac{\left(\EE[|e^{-s_2\beta}|]\right)^{(t-\tau)w}}{1-\left(\EE[e^{s_2\alpha}]\cdot\EE[e^{-s_2\beta}]\right)^{t-\tau}},
\end{aligned}
\end{equation}
since $0 \leq \left(\EE[e^{-s_i\beta}]\right)^{(t-\tau)w} \leq 1$ and the Mellin transform of the service is decreasing within the stability interval $(0,b)$, reaching the maximal value of 1 for \mbox{$s_i=0, i\in \{1,2\}$} and \mbox{$(t-\tau)w \in \mathbb{Z}$}, guaranteed by the discrete time-domain assumption done in Sec.~\ref{sec:snc_basics}. It follows that Eq.~\eqref{eq:whole_convexity} holds which concludes the delay bound convexity proof. 

To prove the second part of Theorem~\ref{thm:2}, we observe that Theorem~\ref{thm:1} defines  $\K^{\mathbb{L}} (s, -w)$ in terms of a recursion starting with the single hop kernel, which is convex in $s$ according to the first part of the proof above. As \mbox{$\M_{g(\gamma)}(1-s)>0, \forall s>0$}, the theorem follows since any positive linear combination of convex functions is also convex \cite{Book:clarke}. 
\end{proof}

\end{appendices}

\bibliographystyle{IEEEtran}
\bibliography{mybib}
              
\end{document}